\pdfoutput=1
\documentclass{article}

\usepackage{microtype}
\usepackage{graphicx}
\usepackage{subfigure}
\usepackage{booktabs} 
\usepackage{enumitem}

\usepackage{hyperref}

\usepackage[accepted]{icml2024}

\usepackage{amsmath}
\usepackage{amssymb}
\usepackage{mathtools}
\usepackage{amsthm}
\allowdisplaybreaks

\usepackage[capitalize,noabbrev]{cleveref}

\DeclarePairedDelimiter{\dotp}{\langle}{\rangle}

\DeclareMathOperator*{\argmax}{arg\,max}

\newcommand{\norm}[1]{\left\lVert#1\right\rVert}

\global\long\def\cS{\mathcal{S}}
\global\long\def\cA{\mathcal{A}}

\global\long\def\cX{\mathcal{X}}

\global\long\def\cM{\mathcal{M}}
\global\long\def\cZ{\mathcal{Z}}

\global\long\def\cR{\mathcal{R}}

\global\long\def\fG{\mathfrak{G}}

\global\long\def\E{\mathbb{E}}
\global\long\def\R{\mathbb{R}}

\global\long\def\KL{\text{KL}}
\global\long\def\TV{\text{TV}}
\global\long\def\ER{\overline{\mathbb{R}}}

\global\long\def\rgap{\text{RNEGap}}

\global\long\def\ne{\textup{NE}}
\global\long\def\dom{\operatorname{dom}}
\global\long\def\core{\operatorname{core}}
\global\long\def\Id{\operatorname{Id}}
\global\long\def\Var{\operatorname{Var}}

\global\long\def\wover{\overline{\omega}}

\global\long\def\gover{\overline{g}}

\global\long\def\V{\overline{V}}
\global\long\def\Q{\overline{Q}}

\global\long\def\VR{\widetilde{V}}
\global\long\def\QR{\widetilde{Q}}
\global\long\def\VRT{\widehat{V}}

\global\long\def\rbr{\ensuremath{\mathsf{br}}}

\global\long\def\pbei{\widehat{P}}
\global\long\def\brbei{\widehat{\boldsymbol{r}}}
\global\long\def\rbei{\widehat{r}}
\global\long\def\gbei{\widehat{G}}

\global\long\def\ba{\boldsymbol{a}}
\global\long\def\br{\boldsymbol{r}}
\global\long\def\bpi{\boldsymbol{\pi}}
\global\long\def\bpine{\boldsymbol{\pi}^{\dagger}}
\global\long\def\pine{\pi^{\dagger}}

\global\long\def\bigO{\mathcal{O}}

\global\long\def\norm#1{\left\Vert #1\right\Vert }

\global\long\def\brno{\boldsymbol{r}^{\star}}
\global\long\def\rno{r^{\star}}
\global\long\def\pno{P^{\star}}
\global\long\def\gno{G^{\star}}
\global\long\def\rset{\mathcal{R}}
\global\long\def\pset{\mathcal{P}}
\global\long\def\uset{\mathcal{U}}

\global\long\def\urset{\mathcal{U}^{r}}

\global\long\def\argmax{\operatorname*{arg\,max}}

\global\long\def\argsup{\operatorname*{arg\,sup}}
\global\long\def\arginf{\operatorname*{arg\,inf}}

\global\long\def\int{\text{int}}
\global\long\def\Bd{\text{Bd}}

\newcommand{\Real}{\mathbb{R}}

\theoremstyle{plain}
\newtheorem{theorem}{Theorem}[section]
\newtheorem{proposition}[theorem]{Proposition}
\newtheorem{lemma}[theorem]{Lemma}
\newtheorem{corollary}[theorem]{Corollary}
\theoremstyle{definition}
\newtheorem{definition}{Definition}[section]
\newtheorem{assumption}{Assumption}[section]
\theoremstyle{remark}

\newtheorem*{remark*}{Remark}
\usepackage[textsize=tiny]{todonotes}

\icmltitlerunning{Robustness and Regularization in Markov Games}

\begin{document}

\twocolumn[
\icmltitle{Roping in Uncertainty: 
Robustness and Regularization in Markov Games} 





\begin{icmlauthorlist}
\icmlauthor{Jeremy McMahan}{uw}
\icmlauthor{Giovanni Artiglio}{uw}
\icmlauthor{Qiaomin Xie}{uw}
\end{icmlauthorlist}

\icmlaffiliation{uw}{University of Wisconsin-Madison, USA}

\icmlcorrespondingauthor{Jeremy McMahan}{jmcmahan@cs.wisc.edu}

\icmlkeywords{Machine Learning, ICML}

\vskip 0.3in
]



\printAffiliationsAndNotice{}  

\begin{abstract}
We study robust Markov games (RMG) with $s$-rectangular uncertainty. We show a general equivalence between computing a robust Nash equilibrium (RNE) of a $s$-rectangular RMG and computing a Nash equilibrium (NE) of an appropriately constructed regularized MG. The equivalence result yields a planning algorithm for solving $s$-rectangular RMGs, as well as provable robustness guarantees for policies computed using regularized methods. However, we show that even for just reward-uncertain two-player zero-sum matrix games, computing an RNE is PPAD-hard. Consequently, we derive a special uncertainty structure called efficient player-decomposability and show that RNE for two-player zero-sum RMG in this class can be provably solved in polynomial time. This class includes commonly used uncertainty sets such as $L_1$ and $L_\infty$ ball uncertainty sets. 
\end{abstract}

\section{Introduction}

Offline reinforcement learning (RL) and RL in simulated environments are effective ways to deal with situations where traditional online RL would be too risky or costly. However, these approaches suffer from the sim-to-real gap in which slight differences in the models can lead to policies with poor performance in the true environment. To combat the sim-to-real gap, robust policies were studied using the framework of robust Markov decision processes (RMDP)~\cite{RMDP-original} and later robust Markov Games (RMG)~\cite{RMG-original}. Robust approaches have been effective in the real world, especially for navigating UAVs in mission-critical multi-agent environments \cite{RMARL-uav-application} and in queuing systems \cite{RMARL-queuing}. In practice, regularization has been a popular approach to improving the robustness and convergence of multi-agent RL algorithms with empirical success. 

A robust Markov game $(\cS,\{\cA_{i}\}_{i\in[N]},\pno,\brno,H,\uset)$ is defined by a standard Markov game (MG) $\gno:=(\pno,\brno)$, called the nominal game, and an uncertainty set $\uset:=\pset\times\urset$ that describes all possible models that could be realized. Here, $\pset$ is the set of possible transition kernels, and $\urset$ is the set of possible reward functions. A common solution concept for RMGs is the 
Markov-perfect robust Nash Equilibrium (MPRNE). A policy $\bpi$ is an MPRNE if, for each stage game $(h,s)$, $\bpi$ is a mutual best response assuming the worst-case model for each player.

Unlike classical RMDPs, solving RMGs is already difficult when only the reward function has uncertainty. Specifically, a single-stage game with reward uncertainty can capture arbitrary general-sum games and so is PPAD-hard to solve. Surprisingly, unlike traditional game theory, even the two-player zero-sum version of reward-uncertain RMGs with $|\cS| = H = 1$ and the minimal $(s,a)$-rectangularity assumption is PPAD-hard to solve. Similarly, two-player zero-sum RMGs with only transition uncertainty and $H=2$ periods are also PPAD-hard to solve. Thus, solving even simple RMGs is already a computational challenge.

Although many advances have been made for RMDPs, RMGs are much less understood. The seminal paper~\cite{RMG-original} devised algorithms to learn a robust NE (RNE) for RMGs but only proved asymptotic convergence of their methods. On the other hand, \citet{RMARL-double-pessimism} proposed provably sample-efficient algorithms to learn an RNE for the special case of $(s,a)$-rectangular RMGs, but their methods require an efficient planning oracle that does not currently exist in the literature. Creating such a planning oracle is one of the goals of this work. Lastly, adding a regularizer to the value function of an MG has shown promise to improve robustness empirically, but formal guarantees have not been shown in the multi-agent setting~\cite{efficient-regularizer}.

\paragraph{Our Contributions.} We study the computational complexity of computing MPRNE for RMGs with $s$-rectangular uncertainty. We show that computing an MPRNE of an RMG with $s$-rectangular uncertainty can be done by computing a Markov-perfect NE (MPNE) of an appropriately designed regularized MG. In particular, the regularizer corresponds to the support function of the uncertainty set for the stage game. 
Furthermore, we show that for most well-known regularization functions, such as entropy and $\ell_p$ norm regularizers, the set of MPNE for a regularized game corresponds exactly to the set of MPRNE for an RMG with an interpretable uncertainty. This fact implies that for common classes of regularizers and uncertainty sets, the problems of solving RMGs and regularized MGs are polynomial-time equivalent. Thus, any efficient off-the-shelf algorithm for regularized MGs can be used to efficiently compute robust policies, confirming the empirical phenomenon mathematically.

We also show that many useful classes of RMGs with $s$-rectangular reward uncertainty can be solved in polynomial time. As in classical game theory, our first step is to extend the notion of zero-sum games to the robust setting. Although the zero-sum property does not guarantee efficiency, our proof of computational hardness reveals a key structural bottleneck to efficiency: general-sum behavior is simulated whenever the support function output of the reward uncertainty set involves a product of each player's policy. In contrast, if the support function decomposes into separate parts for both players, $\sigma(\bpi) = \Omega_1(\pi_1) + \Omega_2(\pi_2)$, we show the equivalent regularized MG is also zero-sum. Thus, our planning algorithm runs in polynomial time so long as the uncertainty satisfies what we call the \emph{efficiently player-decomposable assumption}. This assumption permits many standard sets such as $L_1$ and $L_\infty$-ball uncertainty sets.

\subsection{Related Work}
\paragraph{Robust MDPs.} Robust MDPs have been studied under many different uncertainty structures. The original structure, called $(s, a)$-rectangularity, was first introduced in \cite{RMDP-original, RMDP-sa-rect}. Many attempts to generalize $(s, a)$-rectangularity have led to a rich family of rectangularity notions including $s$-rectangularity ~\cite{RMDP-s-rect-formulation, RMDP-s-rect}, $r$-rectangularity ~\cite{RMDP-r-rect-formulation, RMDP-r-rect}, $k$-rectangularity ~\cite{RMDP-k-rect}, and $d$-rectangularity ~\cite{RMDP-d-rectangularity}. Many standard MDP techniques have also been extended to the robust setting including dynamic programming~\cite{RMDP-RDP, RMDP-fast-bellman}, policy iteration~\cite{RMDP-RMPI, RMDP-efficient-L1}, policy gradient~\cite{RMDP-PG, RMDP-PG-convergence}, and function approximation~\cite{RMDP-kernel-based, RMDP-func-approx}. Regularized MDP techniques also successfully solve robust MDPs due to a general equivalence for many uncertainty sets~\cite{RMDP-regularize-robustness-original, RMDP-regularize-robustness, RMDP-Regularization} including both $(s, a)$ and $s$-rectangularity.

In the learning setting, standard RL approaches have been successfully ``robustified'' including model-based approaches~\cite{RRL-online}, Q-learning~\cite{DRRL-Q-learning}, policy gradient~\cite{RRL-PG, RRL-LSPG}, and kernel methods~\cite{RMDP-kernel-based}. Strong theoretical results have also pinned down the sample complexity of many methods~\cite{RRL-sample-complexity, DRRL-offline-pac, RRL-asymptotics}. In fact, \citet{RRL-adversarial} showed that robust learning is equivalent to learning in adversarial games and this is further exploited using game-theoretical techniques \cite{RM-RNE}.

\paragraph{Robust MGs.} Robust normal form games were first introduced by \citet{RM-robust-game-theory}. \citet{RM-gen-sum-reduction} showed that robust games can be reduced to general sum games, but computationally efficient methods have yet to be established. The notion of robustness has been extended to Markov games \cite{RMG-original, RMARL-queuing}. A sample efficient approach for learning robust policies under $(s,a)$-rectangularity is derived by \citet{RMARL-double-pessimism}, but their method relies on a planning oracle that has yet to be derived in the literature. Our work provides the efficient planning oracle needed to make those methods tractable and extends beyond just $(s, a)$-rectangularity. In contrast, ~\citet{RMARL-v-learning} addresses the problem of learning a robust CCE with low sample complexity whereas we focus on computing the stronger solution concept of robust NE.

\paragraph{Regularization in MDP and MGs.} Various regularization methods have been extensively used in MDPs \cite{RMDP-PG, geist2019regularized} and games~\cite{grill2019entropy,cen2021game_entropy,zhang2023learning_graphon,mertikopoulos2016learning}, with diverse motivations, such as improved exploration~\cite{lee2018tsallis}, stability~\cite{schulman2017proximal} and convergence~\cite{cen2021game_entropy,zhan2023PMD_regularized}. Popular regularizers include a variety of entropy functions and KL divergence. Recent works relate regularization to robustness in MDP/RL~\cite{brekelmans2022regularizer,eysenbach2021maximum,husain2021regularized}. In particular, \citet{RMDP-regularize-robustness} provides an equivalence between regularization and robustness in MDPs. However, the robustness-regularization duality is much less understood in games. Our work fills this gap and opens the path to efficient planning and learning algorithms for achieving robustness in games via regularization. 

\paragraph{Notation.} For an integer $N,$ we denote $[N]:=\{1,2,\ldots,N\}.$
We define the extended reals by $\ER=\R\cup\{-\infty,\infty\}$. For a given finite set $\cZ$, we denote by $\Delta(\cZ)$ the probability simplex
over $\cZ$, and denote by $\R^{\cZ}$ the class of real-valued functions
defined over $\cZ$. For a set $\cM\subset\R^{\cZ}$, the characteristic
function $\delta_{\cM}:\R^{\cZ}\rightarrow\ER$ is defined as $\delta_{\cM}(x)=0$
if $x\in\cM$ and $+\infty$ otherwise. The Legendre-Fenchel transform
of $\delta_{\cM}$ is the so-called support function $\sigma_{\cM}:\R^{\cZ}\rightarrow\ER,$
with $\sigma_{\cM}(y):=\sup_{x\in\cM}\dotp{x,y}$. For a compact set $\mathcal{W}\subset \R^n$, we denote the interior of $\mathcal{W}$ by $\int(\mathcal{W})\subset \mathcal{W},$ and the boundary of $\mathcal{W}$ by $\Bd(\mathcal{W})= \mathcal{W}\setminus \int(\mathcal{W}).$

With slight overload of notation, $\langle\cdot,\cdot\rangle$ denotes the standard inner product when the inputs are vectors, and denotes the Frobenius (also known as component-wise) inner product when the inputs are matrices of the same dimensions.
We will overload the functions such as $\log(\cdot)$ and $\exp(\cdot)$
to take vector inputs, meaning the function is applied in an entry-wise
manner. That is, for a vector $x=(x_{1},\ldots,x_{n})^{\top}\in\R^{n},$
the notation $\exp(x):=(\exp(x_{1}),\ldots,\exp(x_{n}))^{\top}\in\R^{n}$. For a matrix $B\in\R^{m\times n}$, we denote by $\norm B_{p\rightarrow q}=\sup\big\{ \norm{Bx}_{q}:x\in\R^{n}\text{ with }\norm x_{p}=1\big\} $ the operator norm on the space $\R^{m\times n},$ where $\norm{\cdot}_{p}$ denotes the vector $\ell_{p}$-norm. 
The dual to a norm $\norm{\cdot}$
is defined as $\norm v_{*}=\sup_{\norm u\leq1}\dotp{u,v}.$ For any
number $p\in[0,\infty]$, we use $p^{*}\in[0,\infty]$ to denote its
conjugate satisfying $\frac{1}{p}+\frac{1}{p^{*}}=1$. Therefore,
the dual norm of $\left\Vert \cdot\right\Vert _{p}$ is $\left\Vert \cdot\right\Vert _{p^{*}}.$

\section{Preliminaries}\label{sec: preliminaries}

In this work, we consider $H$-horizon Markov games (MG) of $N$ players, with a finite state space $\cS,$ and a finite action space $\cA_{i}$ for each player $i\in[N]$. We let $\cA:=\cA_{1}\times\cA_2\cdots\times \cA_N$ denote
the joint action space and let $\ba=(a_{1},\ldots,a_{N})\in\cA$ denote a joint action of $N$ players. Without loss of generality, we assume that the initial state $s_{1}$ is fixed
\footnote{For the case that the initial state is stochastic, one can add $s_{0}$
as the initial state and set the initial state distribution as the
transition kernel from $s_{0}$ to $s_{1}$.}.
For such a MG, we use $G=(P,\br)$ to represent the game model, where $P=\{P_{h}:\cS\times\cA\rightarrow\Delta(\cS)\}_{h\in[H]}$ is the transition kernel, 
and $\br=\{r_{i,h}\}_{i\in[N],h\in[H]}$
is the deterministic reward function, with $r_{i,h}(s,\ba)$
being the reward for player $i\in[N]$ given that the joint action $\ba$ is applied at state
$s$ at step $h$.

A Markovian policy for player $i$ is denoted by $\pi_{i}=\{\pi_{i,h}:\cS\rightarrow\Delta(\cA_{i})\}_{h\in[H]}$,
with $\pi_{i,h}(\cdot|s)$ being the the strategy of player $i$ at
state $s$ at step $h$. We use $\bpi:=(\pi_{1},\ldots,\pi_{N})$
to represent a \emph{product} joint policy of the $N$ players.
For each player $i\in[N]$, $\bpi_{-i}$ denotes the joint policy
of all players except player $i$.
Let $\Pi_{i}$ be the set of all Markov policies for player $i,$
and $\Pi:=\Pi_{1}\times\Pi_{2}\times\cdots\times\Pi_{N}$ be
the set of all product Markov joint policies of the $N$ players, and
$\Pi_{-i}=\Pi_{1}\times\cdots\times\Pi_{i-1}\times\Pi_{i+1}\cdots\times\Pi_{N}$. We overload the notation $\Delta(\cA):=\Delta(\cA_1)\times \cdots \Delta(\cA_N)$ to represent the space of all product joint policy of the $N$ players at every single state. 

\subsection{Robust Markov Game}

\paragraph{Robust solution concept.} Recall that a robust Markov game (RMG) $(\cS,\{\cA_{i}\}_{i\in[N]},\pno,\brno,H,\uset)$ is defined by a standard MG $\gno:=(\pno,\brno)$, called the nominal game, and an uncertainty set $\uset:=\pset\times\urset$ that describes all possible models. Here, $\pset$ is the set of possible transition kernels, and $\urset$ is the set of possible reward functions. 
Given a product joint policy $\bpi$, for each player $i\in[N]$ we define the \emph{robust value functions }of $\bpi$ with respect
to the uncertainty set $\uset$ as follows: $\forall s\in\cS,\mu\in\Delta(\cA),h\in[H],$
\begin{align}
V_{i,h}^{\bpi}(s) & :=\inf_{G\in\uset}\V_{i,h}^{\bpi}(s,G),\label{eq:V_RMG}\\
Q_{i,h}^{\bpi}(s,\mu) & :=\inf_{G\in\uset}\E_{\ba\sim\mu}\left[\Q_{i,h}^{\bpi}(s,\ba,G)\right],\label{eq:Q_RMG}
\end{align}
where $\V_{i,h}^{\bpi}(\cdot,G)$ and $\Q_{i,h}^{\bpi}(\cdot,\cdot,G)$
are the standard value functions for the MG $G=(P,\br)$ (cf.\ equations \eqref{eq:V_MG}-\eqref{eq:Q_MG} in Appendix \ref{sec:proof_section_RMG}).
Similar to the Bellman equation for standard MG, we have a robust Bellman equation, as stated in Proposition \ref{prop:Bellman}. 

We let $V_{i,h}^{\dagger,\bpi_{-i}}(s)$ denote the optimal value
for player $i$ starting from state $s$ at step $h$, given a product
Markov joint policy $\bpi_{-i}$ of all players except player $i$,
i.e., 
\begin{equation}
V_{i,h}^{\dagger,\bpi_{-i}}(s)=\sup_{\pi'_{i}\in\Pi_{i}}V_{i,h}^{\pi'_{i}\times\bpi_{-i}}(s),\;\;\forall s\in\cS,h\in[H].\label{eq:V_BR}
\end{equation}
A policy $\pi_{i}$ that attains the optimal value $V_{i,h}^{\dagger,\bpi_{-i}}(s),$ for all $s\in\cS,h\in[H]$ is a robust best response policy to a given $\bpi_{-i}.$ 
It is well-known that when $\bpi_{-i}$ is Markovian, the best response amongst all history-dependent policies is Markovian, as it reduces to solving a single-agent robust MDP problem \cite{RMDP-RDP}. Our work will focus on Markov policies $\pi_{i}$ as above.
Compared with the best response policy in standard MGs, the robust best response policy maximizes the \emph{robust }value function of each player $i$ given other players' policy $\bpi_{-i}$, leading to the notion of \emph{robust Nash equilibrium} \cite{RMG-original,RMARL-double-pessimism}. 

\begin{definition}
\label{def:RNE} A joint product policy $\bpi=\{\bpi_{h}\}_{h\in[H]}$ is a \emph{Markov perfect robust Nash equilibrium
}if it holds that 
\[
V_{i,h}^{\bpi}(s)=V_{i,h}^{\dagger,\bpi_{-i}}(s), \qquad\forall i\in[N], h \in [H], s \in \cS.
\]
\end{definition}

{\bf Rectangularity.} It is common to choose the uncertainty
set centered around the nominal model $\gno:=(\pno,\brno)$.
In this work, we consider reward uncertainty sets of the form $\urset=\brno+\rset$. 
We allow the reward uncertainty sets to potentially depend on
the players' policy $\bpi\in\Pi$, denoted
by $\urset(\bpi)=\brno+\rset(\bpi).$ When the context is clear, we
will drop the notation $(\bpi)$ for the ease of exposition. 
As the robust MDP literature \cite{RMDP-s-rect, RMDP-regularize-robustness, RMDP-sa-rect}, we consider uncertainty sets that satisfy certain \emph{rectangular condition}. 

\begin{definition}
\label{assu:rectangular}[$s$-rectangular Uncertainty Set] The uncertainty
sets $\uset:=\pset\times\urset$ with $\urset=\brno+\rset$ satisfy
\begin{align*}
\pset &=\times_{(s,h)\in\cS\times [H]}\pset_{s,h},\quad
\uset^{r} =\times_{(i,s,h)\in[N]\times\cS\times [H]}\uset_{i,s,h}^{r},
\end{align*}
where $\pset_{s,h}\subset \R^{\cS\times \cA}$ and $\uset_{i,s,h}^{r}=\rno_{i,h}(s,\cdot)+\rset_{i,s,h}\subset\R^{\cA}$ are closed, convex sets.
\end{definition}

\begin{definition}
\label{assu:(s,a)-rectangular}[$(s,a)$-Rectangular Uncertainty Set] A special case of $s$-rectangularity is a $(s,a)$-rectangular uncertainty set:
\begin{align*}
\pset & =\times_{(s,\ba,h)\in\cS\times\cA\times [H]}\pset_{s,\ba,h}, \\
\uset^{r} &=\times_{(i,s,\ba,h)\in[N]\times\cS\times \cA \times [H]}\uset_{i,s,\ba,h}^{r},
\end{align*}
where $\pset_{s,\ba,h}\subset\Delta(\cS)$ and $\uset_{i,s,\ba,h}^{r}=\rno_{i,h}(s,\ba)+\rset_{i,s,\ba,h}\subset\R$ are closed, convex sets.
\end{definition}

We remark that the state-action
value function defined in \eqref{eq:Q_RMG} is different from the classical
setting. The new definition here allows us
to provide a unified framework for both $s$-rectangular and $(s,a)$-rectangular uncertainty sets.

{\bf Robust suboptimality gap. }
Given a joint policy $\bpi$, for each player $i\in[N]$, we define the robust suboptimality gap at step $h$ as
\begin{equation}
\rgap_{i,h}(\bpi,s):=V_{i,h}^{\dagger,\bpi_{-i}}(s)-V_{i,h}^{\bpi}(s).\label{eq:RNEgap-1}
\end{equation}
That is, the RNE gap measures the suboptimality gap of player $i$'s
policy $\pi_{i}$ against its robust best response policy given all
other players' policy $\bpi_{-i}.$ It is clear that $\rgap_{i,h}(\bpi,s)\geq0$
for all product joint $\bpi$ and all $s\in\cS$ and $h\in [H]$. For any MPRNE policy
$\bpi^{*}$, we have $\rgap_{i,h}(\bpi^*,s)=0$.

\paragraph{Existence of RNE.} 
Recall that the reward uncertainty sets can depend on the players' policy
$\bpi\in\Pi$, denoted by $\urset(\bpi).$
We consider the following assumption on the reward uncertainty sets,
which ensures the existence of robust Nash equilibrium. 
We first define the continuity of a point-to-set function. Consider
a function $f$ that maps each $z\in\R^{d_{1}}$ to a set in $\R^{d_{2}}$.
The function $f$ is continuous at $z\in\R^{d_{1}}$ if it
satisfies the following: for each $\epsilon>0$, there exists $\delta(\epsilon)>0$
such that for all $z'\in\R^{d_{1}}$ with $\left\Vert z-z'\right\Vert _{\infty}<\delta(\epsilon),$ it holds that
$D\big(f(z),f(z')\big):=\max\big\{\sup_{y\in f(z)}\inf_{y'\in f(z')}\left\Vert y-y'\right\Vert _{\infty},\\ \sup_{y'\in f(z')}\inf_{y\in f(z)}\left\Vert y-y'\right\Vert _{\infty}\big\}<\epsilon.$

\begin{assumption}
\label{assu:reward} Consider $s$-rectangular uncertainty set. For each $\bpi\in\Pi,$ the reward uncertainty
set $\uset^{r}(\bpi)=\times_{(i,s,h)\in[N]\times\cS\times [H]}\uset_{i,s,h}^{r}(\bpi),$
where $\uset_{i,s,h}^{r}(\bpi)=\rno_{i,h}(s,\cdot)+\rset_{i,s,h}(\bpi)\subset\R^{\cA}$
satisfies the following for all $(i,s,h)\in[N]\times\cS\times H$:
\begin{enumerate}
\item Bounded game value: There exists a constant $L_{r}$
such that for each $r\in\uset_{i,s,h}^{r}(\bpi),$ it holds that $L_{r}\leq\E_{\ba\sim\bpi(\cdot\mid s)}\left[r(\ba)\right].$
\item Convexity: The support function $\sigma_{\rset_{i,s,h}(\bpi)}:\R^{\cA}\rightarrow(-\infty,+\infty]$
of $\rset_{i,s,h}(\bpi)$, defined as $\sigma_{\cR_{i,s,h}(\bpi)}(y):=\sup_{x\in\cR_{i,s,h}(\bpi)}\left\langle x,y\right\rangle $,
satisfies that $\sigma_{\rset_{i,s,h}(\bpi)}\big(-\pi_{i,h}(s)\bpi_{-i,h}^{\top}(s)\big)$
is convex in $\pi_{i,h}(s)$ and continuous at all $\bpi\in\Bd(\Pi).$
\item Continuity: The set $\rset_{i,s,h}(\cdot)$ is continuous at all $\bpi\in\int(\Pi)$;
and $\sup_{y\in\rset_{i,s,h}(\bpi)}\left\Vert y\right\Vert _{\infty}<\infty$
for all $\bpi\in\int(\Pi)$.
\end{enumerate}
\end{assumption}

We remark that a special reward uncertainty set that is policy-independent and bounded satisfies Assumption \ref{assu:reward} \cite{RM-robust-game-theory}.
Such uncertainty set has been considered in discounted Markov games \cite{RMARL-queuing,RMG-original}. In this paper, we consider
more general reward uncertainty sets. Importantly, the boundedness
condition in Assumption \ref{assu:reward} does not require the uncertainty
set to be uniformly bounded. This relaxation allows us to consider
uncertainty sets based on log-barrier functions.

The following theorem states that a robust Nash equilibrium (RNE) always exists. The proof is provided in Section~\ref{sec:proof_existence_rne}. A similar existence result has been proved for the case with only $(s,a)$-rectangular transition uncertainty \citep{RMARL-double-pessimism}.

\begin{theorem}
\label{thm:Existence_RNE}(Existence of RNE). Given a RMG $(\cS,\{\cA_{i}\}_{i\in[N]},\pno,\brno,H,\uset)$
with $s$-rectangular uncertainty set $\uset$ satisfying Assumption \ref{assu:reward}, the robust
Nash equilibrium defined in Definition (\ref{def:RNE}) always exists.
Moreover, a joint policy $\bpine=\{\bpine_{h}\}_{h\in[H]}$ defined
as follows is an MPRNE: 
\[
\bpine_{h}(\cdot\mid s)\in\ne\left(\{Q_{i,h}^{\bpine}(s,\cdot)\}_{i\in[N]}\right),\forall s\in\cS,h\in[H],
\]
where $\ne(\cdot)$ denotes the Nash equilibrium of a general-sum,
normal-form game. 
\end{theorem}

\subsection{Regularized Markov Games}

A regularized Markov game with $N$ players can be described by a
tuple $(\cS,\{\cA_{i}\}_{i\in[N]},P,\br,H,\Omega),$ with $G=(P,\br)$ being a standard MG model. Here $\Omega:=\left(\Omega_{i,h}\right)_{i\in[N],h\in[H]}$
is a set of policy regularization functions such that for all
$(i,h)\in[N]\times[H]$, $\Omega_{i,h}:\cS\times\Delta(\cA)\rightarrow\R$
is convex in $\pi_{i,h}$ given a fixed $\bpi_{-i}$. Given a joint policy $\bpi\in\Pi$, the \emph{regularized value functions} for each player $i$
are defined as follows: $\forall s\in\cS,\ba\in\cA,h\in[H],$
\begin{align*}
& \VR_{i,h}^{\bpi}(s,G)=\E_{P}^{\bpi}\Big[\sum_{t=h}^{H}r_{i,t}(s_{t},\ba_{t})-\Omega_{i,t}(s_t,\bpi_{t})|s_h=s\Big],\\
& \QR_{i,h}^{\bpi}(s,\ba,G)=r_{i,h}(s,\ba)+\\
&\E_{P}^{\bpi}\Big[\sum_{t=h+1}^{H}r_{i,t}(s_{t},\ba_{t})-\Omega_{i,t}(s_{t},\bpi_{t})|s_{h}=s,\ba_{h}=\ba\Big].\nonumber 
\end{align*}
A common solution concept for regularized MG is the \emph{Markov perfect Nash equilibrium} (MPNE). 
\begin{definition}
\label{def:RE_NE} A joint product policy $\bpi=\{\bpi_{h}\}_{h\in[H]}$ is an MPNE for a regularized MG if it holds that 
$\VR_{i,h}^{\bpi}(s)=\sup_{\pi'_{i}\in\Pi_{i}} \VR_{i,h}^{\pi'_{i}\times\bpi_{-i}}(s), \forall i\in[N], h \in [H], s \in \cS.
$\end{definition}

\section{Markov Games with Reward Uncertainty}\label{sec:reward_robust}
In this section, we consider robust Markov games with only reward uncertainty. 
To solve such RMGs, we take inspiration from the single-player setting. We follow the ideas of \citet{RMDP-regularize-robustness} to show that solving RMGs can be done through solving regularized MGs. In fact, for many common regularizers, the reverse is also true, which implies \emph{regularization yields robust solutions}. Importantly, the equivalence results mean existing efficient, regularized MG solvers can be used off-the-shelf to efficiently solve robust MGs. Proofs for the results in this section are deferred to \cref{sec:proof_section_reward}.

\subsection{Matrix Games}\label{sec:reward_matrix}
To build intuition for our results, we first study matrix games with $N$ players since they can be viewed as a simple Markov game with $S = H = 1$. Already through matrix games, we see that reward uncertainty is much more complex to handle than in the single-agent setting. Nevertheless, by understanding the support functions induced by uncertainty sets, we can derive equivalent regularized matrix games.

\paragraph{Reward Structure.} We consider $s$-rectangular reward uncertainty set of the form $\uset=\brno+\cR$, where $\brno\in\uset$
is the nominal model, and $\rset=\times_{i\in[N]}\cR_{i}$ with $\rset_{i}\subset\R^{\cA}.$
Recall that the uncertainty set $\rset$ can potentially depend on the players' joint
policy $\bpi\in \Delta(\cA)$. We observe that the robust value for player $i$ can be simplified to 
\begin{equation*}
    V_{i}^{\bpi}=\inf_{\br\in\uset}\E_{\ba\sim\bpi}\left[r_{i}(\ba)\right]=\inf_{\br\in\uset}\pi_{i}^{\top}r_{i}\bpi_{-i}.
\end{equation*}
Here, we view $r_{i}$ as a matrix in $\R^{\cA_{i}\times\cA_{-i}}$,
$\pi_{i}$ as a column vector in $\R^{\cA_{i}}$ and $\bpi_{-i}$
as a column vector in $\R^{\cA_{-i}}$. 
By Proposition \ref{prop:reward_matrix}, we can equivalently rewrite the robust value function as follows:
\begin{align}
 V_{i}^{\bpi}&= \pi_{i}^{\top}\rno_{i}\bpi_{-i}-\sigma_{\rset_{i}}(-\pi_{i}\bpi_{-i}^{\top}), \label{eq:v_pi_sigma}
\end{align}
where $\sigma_{\cR_{i}}(y) = \sup_{x\in\cR_{i}}\langle x,y\rangle$ denotes the support function of $\cR_{i}$. 

\paragraph{Equivalence between Robustness and Regularization.} If we consider the support function $\sigma_{\cR_{i}}$ as a regularizer,
equation \eqref{eq:v_pi_sigma} shows an equivalence between player
$i$'s robust game value $V_{i}^{\bpi}$ and the value $\VR_{i}^{\bpi}(s,\gno)$ in a $\sigma_{\cR_{i}}$-regularized game. Note that the RNE
of the matrix game is a product joint policy $\bpine$ such
that $\pine_{i}$ is the robust best response policy to $\bpine_{-i}$. That is,
\begin{equation*}
\pine_{i}\in\argsup_{\pi_{i}\in\Delta(\cA_{i})}V_i^{\pi_{i}\times\bpine_{-i}}=\argsup_{\pi_{i}\in\Delta(\cA_{i})}\VR_{i}^{\pi_{i}\times\bpine_{-i}}(s,\gno),
\end{equation*}
which implies an equivalence between solving a robust game and solving a regularized game.

\begin{theorem}
\label{thm:equiv_reward_matrix} Consider a robust matrix game $\fG$ with uncertainty set $\uset=\brno+\rset$ satisfying Assumption \ref{assu:reward}, where $\brno\in\uset$ is the
nominal model, and $\rset=\times_{i\in[N]}\cR_{i}$. Consider a regularized normal form game $\fG'$ with payoff matrix $\brno$ and the regularizer function $\Omega_i:\Delta(\cA)\rightarrow \R$ for each player $i\in[N]$ defined as $\Omega_i(\bpi):=\sigma_{\rset_{i}}(-\pi_{i}\bpi_{-i}^{\top}).$ Then, $\bpi$ is a RNE of robust game $\fG$ if and only if $\bpi$ is a NE of regularized game $\fG'$. 
\end{theorem}

We provide the proof of Theorem \ref{thm:equiv_reward_matrix} in Appendix~\ref{sec:proof_thm_equiv_reward_matrix}.

\begin{corollary}
    Robust matrix games can be solved using any planning algorithm for regularized games.
\end{corollary}

\paragraph{Interpretable Equivalence.} Thus, we see that we can solve a given robust game by solving a particular regularized game. However, since our reduction maps robust games to very specialized regularized games, it is unclear whether commonly used regularized methods can be used to solve robust games. Fortunately, we can show for many common classes of regularizers, solutions to the regularized game correspond to the solutions of robust games with interpretable uncertainty sets. 

\begin{theorem}\label{thm:example_matrix}
Consider a regularized normal form game $\fG'$ with payoff matrix $\brno$ and the regularizer $\Omega_i:\Delta(\cA)\rightarrow \R$ for each $i\in[N]$.
\begin{enumerate}[leftmargin=*]
\item If $\Omega_i$ is $\ell_{p}/\ell_{q}$-norm regularization, i.e. $\Omega_i(\bpi):=\alpha_{i}\norm{-\pi_{i}}_{p}\norm{\bpi_{-i}}_{q}$ for each player $i\in[N]$, then solving for the NE of $\fG'$ is equivalent to solving for RNE of the robust game $\fG$ with $s$-rectangular ball uncertainty $\uset=\brno+\times_{i\in[N]}\cR_{i}$, where 
$\rset_{i}=\big\{ R_{i}\in\R^{\cA_{i}\times\cA_{-i}}:\norm{R_{i}}_{q^*\rightarrow p}\leq\alpha_{i}\big\} .$
\item If $\Omega_i$ is strongly convex and decomposable with kernel $\omega$, i.e., $\Omega_i(\bpi):=\tau_i\sum_{a_{i}}\pi_i(a_i)\omega_{i}(\pi_{i}(a_{i}))$ for each player $i\in[N]$ with $\tau_i \geq 0$, then solving for a NE of $\fG'$ is equivalent to solving for a RNE of robust game $\fG$ with $(s,a)$-rectangular, policy-dependent uncertainty set $\uset({\bpi})=\brno+\times_{i\in[N],\ba\in \cA}\rset_{i,\ba}(\bpi)$, where 
\begin{align}
\cR_{i,\ba}({\bpi})=&\big[\tau_{i}\omega_{i}(\pi_{i}(a_{i}))+g_{i}(\bpi_{-i}(\ba_{-i})),\nonumber\\
&\;\wover_{i}(\pi_{i}(a_{i}))+\gover_{i}(\bpi_{-i}(\ba_{-i}))\big]\subset \R,\label{eq:rset_regularization}
\end{align}
with 
functions $\omega_{i},\wover_{i}:[0,1]\rightarrow\R$
and $g_{i},\gover_{i}:[0,1]\rightarrow\R$ are continuous.
\end{enumerate}
\end{theorem}

See \Cref{sec:proof_thm_example_matrix} for the proof of \Cref{thm:example_matrix}.
We remark that the $s$-rectangular ball-constrained uncertainty set satisfies Assumption \ref{assu:reward}. The policy-dependent uncertainty set in \eqref{eq:rset_regularization} also satisfies Assumption~\ref{assu:reward} by properly choosing the functions $\omega_{i},\wover_{i},g_{i},\gover_{i}$.
Theorem~\ref{thm:example_matrix} implies that the shape of the reward uncertainty set determines the equivalent regularizer function. For example, a ball-constrained uncertainty set corresponds to norm regularization. Also, the size of the uncertainty set, e.g., the radius parameter $\alpha_i$, determines the magnitude of the regularization factor.

\begin{corollary}
    For any regularizer considered in \cref{thm:example_matrix}, solutions to the regularized game are provably robust.
\end{corollary}

\begin{remark*}
    Observe that as long as the functions for regularization and uncertainty sets in \cref{thm:example_matrix} are efficiently computable, then given either a regularized game or a robust game, the construction of the equivalent game can be done in polynomial time. Therefore, \cref{thm:example_matrix} implies the problem of computing an RNE of a robust game for special classes of uncertainty is \emph{polynomial-time equivalent} to the problem of computing an NE of a regularized game with special classes of regularizers. This means the computational complexity of both problems is the same. In particular, an efficient algorithm for one problem implies an efficient algorithm for the other. 
\end{remark*}

\paragraph{Examples of regularization.} A classical example of decomposable regularizers is the negative Shannon entropy $\Omega_i(\bpi)=\sum_{a_{i}\in\cA_{i}}\pi_{i}(a_{i})\log\pi_{i}(a_{i}).$ Entropy regularization is applied extensively in both single-agent MDP and multi-agent games \cite{RMDP-PG,grill2019entropy,geist2019regularized,zhan2023PMD_regularized,cen2021game_entropy,zhang2023learning_graphon}, and has been shown to accelerate the convergence of many learning algorithms. 
Another example is KL divergence regularizer~\cite{schulman2017proximal} $\Omega_i(\bpi)=\sum_{a_{i}\in\cA_{i}}\pi_{i}(a_{i})\log\frac{\pi_{i}(a_{i})}{\mu_{i}(a_{i})}=d_{\KL}(\pi_{i},\mu_{i})$, where $\mu_{i}\in\Delta(\cA_{i})$ is a given distribution. The literature has also considered the negative Tsallis entropy regularizer \cite{lee2018tsallis} 
$\Omega_i(\bpi)=\frac{1}{2}\big(\norm{\pi_{i}}_{2}^{2}-1\big)=\frac{1}{2}\sum_{a_{i}\in\cA_{i}}\left(\pi_{i}(a_{i})^{2}-\pi_{i}(a_{i})\right).$ Theorem \ref{thm:example_matrix} also applies to other regularizers, such as the Renyi entropy regularization~\cite{mertikopoulos2016learning}; See Section~\ref{sec:example_reward_matrix} for additional discussion.

\subsection{Markov Games}

The intuition for matrix games extends directly to Markov games. The main additional ingredient needed for the analysis is backward induction. By applying \cref{thm:equiv_reward_matrix} to each stage game, we can construct a regularized MG whose MPNE are all MPRNE for the RMG.

We consider Markov games with $s$-rectangular reward uncertainty of the form $\uset=\pno\times\left(\brno+\rset\right).$ Similar to matrix games, we see that each player $i$'s robust value function is equivalent to the value function of a regularized MG, as stated in the following proposition. The proof is provided in \cref{sec:proof_prop_reward_Markov}. For notational simplicity, we define,
\[
[P_{h}V](s,\ba):=\E_{s'\sim P_{h}(\cdot\mid s,\ba)}\left[V(s')\right].
\]

\begin{proposition}
\label{prop:reward_Markov} Suppose the uncertainty is $s$-rectangular. Given any product joint policy $\bpi\in \Pi$, the robust value function $\left\{ V_{i,h}^{\bpi}\right\} _{h\in[H]}$ of each player $i\in[N]$ satisfies:
\begin{align}
V_{i,h}^{\bpi}(s)&=  \E_{\ba\sim\bpi_{h}(s)}\left[\rno_{i,h}(s,\ba)+[\pno_{h}V_{i,h+1}^{\bpi}](s,\ba)\right], \nonumber\\
&\quad-\sigma_{\rset_{i,s,h}}\left(-\pi_{i,h}(s)\bpi_{-i,h}^{\top}(s)\right).\label{eq:reward_Markov}
\end{align}
\end{proposition}

Given the proposition, we can similarly construct a regularized Markov game as in the matrix game case.

\begin{theorem}\label{thm:equiv_reward_Markov}
    Consider a robust MG $\fG$ with $s$-rectangular uncertainty set $\uset=\pno\times\left(\brno+\rset\right)$. Consider
    a regularized MG $\fG' = (\cS,\{\cA_{i}\}_{i\in[N]},\pno,\brno,H,\Omega)$,
    where the regularizer functions $\{\Omega_{i,h}\}_{h\in[H]}$ for each player $i\in[N]$ are given by
    $\Omega_{i,h}(s,\mu):=\sigma_{\rset_{i,s,h}}\big(-\mu_i\mu_{-i}^{\top}\big), \forall s\in\cS,h\in[H],\mu\in \Delta(\cA).$ Then, $\bpi$ is a MPRNE for $\fG$ if and only if $\bpi$ is a MPNE for $\fG'$. 
\end{theorem}

\begin{corollary}
    RMGs with $s$-rectangular reward uncertainty can be solved using any planning or learning algorithm for the equivalent regularized Markov games.
\end{corollary}

\begin{remark*}
    The same results from \cref{thm:example_matrix} extend to the full RMG setting by simply using the same shape of uncertainty set for each stage $(s,h)$, as stated in Theorem \ref{thm:example_Markov} in Appendix~\ref{sec:example_reward_Markov}. 
    In particular, our results apply to the well-known regularized game, including $\ell_p/\ell_q$-norm regularization and decomposable regularizers such as entropy regularization and KL divergence. In the following, we give the example of widely employed entropy regularization.
\end{remark*}

\paragraph{Example: Entropy Regularization.} Consider Shannon entropy regularized Markov game with 
\begin{align*}
    \Omega_{i,h}(s,\mu)=\tau_i\sum_{a_{i}\in\cA_{i}}\mu_{i}(a_i)\log\mu_{i}(a_{i}),\quad\forall \mu\in\Delta(\cA),
\end{align*}
where $\tau_i>0$ denotes regularization factor.
The regularized NE is equivalent to the RNE of a robust MG with $(s,a)$-rectangular policy-dependent uncertainty set $\uset({\bpi})=\pno\times\left(\brno+\rset({\bpi})\right),$ where 
$\rset({\bpi})=\times_{i,s,\ba,h} \rset_{i,s,\ba,h}({\bpi})$ with
\begin{align}
\cR_{i,s,\ba,h}({\bpi})&=\big[\tau_{i}\log\pi_{i,h}(a_{i}|s), 
\wover_{i,s,h}(\pi_{i,h}(a_{i}|s))\big],\nonumber
\end{align}
where the function $\wover_{i,s,h}:[0,1]\rightarrow \R$ is continuous and non-negative to ensure that $\brno\in \uset^{\bpi}.$
The reader can find additional details of other regularizers in Section~\ref{sec:example_reward_Markov}.

\section{Efficient Algorithms for Robust Zero-Sum MG}\label{sec:efficient}

Although the results from \cref{sec:reward_robust} provide insights that allow us to solve RMGs, they do not yield efficient algorithms in general. Our reduction may result in general-sum regularized games which are hard to solve. 
As we will see shortly, even for a two-player zero-sum Markov game with reward uncertainty, computing the RNE is PPAD-hard in general. However, we show that when uncertainty sets satisfy a natural decomposition assumption, then an MPRNE of a two-player zero-sum Markov game can be computed in polynomial time. Proofs for the results in this section are deferred to \cref{sec:proof_section_efficient}. 

First, we extend the notion of zero-sum Markov games to the robust setting.
\begin{definition}[Two-Player Zero-Sum RMG]
    A RMG $\fG$ is a \emph{two-player zero-sum RMG} (TPZS RMG) with reward uncertainty if $N = 2$ and for each $h \in [H]$ and $s \in \cS$, 
    $\brno_{2,h}(s,\cdot) = - \brno_{1,h}(s,\cdot)$.
\end{definition}
While this is a crucial first step, focusing on TPZS RMGs is insufficient to guarantee efficient algorithms. Surprisingly, even with the most restrictive notion of $(s, a)$-rectangularity, TPZS RMGs are PPAD-hard to solve. Observe this stands in sharp contrast to traditional game theory where two-player zero-sum games can be efficiently solved using LP methods.

\begin{theorem}\label{thm: zero-sum-hardness}
    Even restricted to the class of $(s,a)$-rectangular uncertainty sets, computing an RNE of a TPZS RMG is PPAD-hard even for $H = S = 1$. 
\end{theorem}

\begin{proof}[Proof Sketch] 

    We present a poly-time reduction from the problem of computing an NE for a general-sum game to the problem of computing an RNE for a two-player zero-sum robust matrix game with $(s,a)$-rectangular reward uncertainty. Since computing an NE of a general-sum game is PPAD-hard, it then follows that computing an RNE for the aforementioned class of robust matrix games is also PPAD-hard. Let $(A,B)$ be an arbitrary general sum matrix game. 

    The idea is to construct a two-player zero-sum robust matrix game $\fG$ defined by $(\brno, \rset_1, \rset_2)$ so that the set of solutions to both games is the same. First, we choose matrices $\brno$, $R_1$ and $R_2$ satisfying $\brno + R_1 = A$ and $\brno + R_2 = -B$. Then, we choose $\rset_1$ and $\rset_2$ satisfying $-\sigma_{\rset_1}(-\pi_1 \pi_2^{\top}) = \pi_1^{\top} R_1 \pi_2$ and $-\sigma_{\rset_1}(-\pi_1 \pi_2^{\top}) = \pi_1^{\top} R_2 \pi_2$. Consequently, we see the robust suboptimality gap for $\fG$ simplifies to,
\begin{align*}
    & \rgap_1(\bpi)+\rgap_2(\bpi) \\
        =& \max_{\pi_1'\in\Delta(\cA_1)}\big\{ \pi_1'^\top (\rno+R_1)\pi_2\big\} -\pi_1^\top (\rno+R_1)\pi_2 \\
        &+\max_{\pi_2'\in\Delta(\cA_2)}\big\{ \pi_1^\top (-\rno-R_2)\pi_2'\big\} - \pi_1^\top (-\rno-R_2) \pi_2 \\
        =& \max_{\pi_1'\in\Delta(\cA_1)}\big\{ \pi_1'^\top A\pi_2\big\} -\pi_1^\top A \pi_2 \\
        &+\max_{\pi_2'\in\Delta(\cA_2)}\big\{ \pi_1^\top B\pi_2'\big\} - \pi_1^\top B \pi_2.
    \end{align*}
    We see the total suboptimality gap of $\bpi$ for $\fG$ exactly matches the suboptimality gap of $\bpi$ for $(A,B)$. This implies the set of solutions to both games is the same as desired.
\end{proof}

The proof of \cref{thm: zero-sum-hardness} is provided in \cref{sec:proof_zero-sum-hardness}. Thanks to \cref{thm: zero-sum-hardness}, we cannot hope to solve TPZS RMGs even under the simplest $(s,a)$-rectangularity assumption. However, the proof reveals key structural properties of the uncertainty set that lead to hardness. Specifically, we needed the uncertainty sets to satisfy $-\sigma_{\rset_1}(-\pi_1 \pi_2^{\top}) = \pi_1^{\top} R_1 \pi_2$. In particular, observe the support function involves a product of $\pi_1$ and $\pi_2$. This property allowed us to simulate a general sum game through the uncertainty sets. 

One way to avoid such hardness is to consider uncertainty sets that break the support function into two separate pieces: one for $\pi_1$ and one for $\pi_2$. Formally, suppose that $\sigma_{\rset_{i}}(-\pi_{i}\bpi_{-i}^{\top}) = \Omega_{i,i}(\pi_i) + \Omega_{i,-i}(\bpi_{-i})$. Immediately, we see that an uncertainty set satisfying $-\sigma_{\rset_1}(-\pi_1 \pi_2^{\top}) = \pi_1^{\top} R_1 \pi_2$ is no longer possible.
Even better, by inspecting the RNEGap of $\fG$, we see that each of the $\Omega_{i,-i}(\bpi_{-i})$ terms cancel out, leading to the following RNEGap: 
\begin{align*}
    &\max_{\pi_1' \in\Delta(\cA_1)} \left\{ {\pi_{1}'}^{\top}\rno\pi_{2}-\Omega_{1,1}(\pi_1') + \Omega_{2,2}(\pi_2) \right\} \\
    &- \min_{\pi_2' \in\Delta(\cA_2)} \left\{ {\pi_{1}}^{\top}\rno\pi_{2}' - \Omega_{1,1}(\pi_1) +\Omega_{2,2}(\pi_2') \right\}.
\end{align*}

Observe this is exactly the suboptimality gap of $\bpi$ for a TPZS regularized game. Importantly, TPZS regularized games can be solved in polynomial time so long as the regularizer functions are strongly convex \cite{facchinei2003finite,cherukuri2017role}. Thus, as long as the decomposition of $\sigma$ into $\Omega$ functions is known or can be computed efficiently, solving the robust matrix game can also be done efficiently by solving the zero-sum regularized game.

\begin{definition}[Efficiently Player-Decomposable]\label{assu:efficient}
    Suppose $\fG$ is a TPZS RMG with $s$-rectangular reward uncertainty. We say that $\fG$ is \emph{efficiently player-decomposable} if $\forall i \in \{1,2\}, s \in \cS, h \in [H], \mu \in \Delta(\cA)$, $\sigma_{\rset_{i,s,h}}(-\mu_i \mu_{-i}^{\top}) = \Omega_{i,i}^h(\mu_i) + \Omega_{i,-i}^h(\mu_{-i}),$ where each $\Omega$ is a strongly convex function that is either known in advance or can be computed from $\fG$ in polynomial time.
\end{definition}

\begin{lemma}\label{lem:efficient-decomp}
    Suppose $\fG$ is a robust matrix game that satisfies \cref{assu:efficient} and $\fG'$ is the corresponding regularized game constructed in \cref{thm:equiv_reward_matrix}. Then, $\fG'$ is a TPZS regularized game with separate regularizers, $\Omega_i = \Omega_{i,i}$ for each player $i$. Therefore, computing an RNE for $\fG$ can be done in polynomial time by solving $\fG'$.
\end{lemma}

The proof of \Cref{lem:efficient-decomp} can be found in \Cref{sec:proof_efficient-decomp}. \Cref{lem:efficient-decomp} implies we can solve each stage game so long as the decomposition holds for each stage. Thus, to solve a zero-sum RMG, we can solve each stage's regularized game using backward induction. 

\begin{theorem}\label{thm:efficient-mg}
    Suppose $\fG$ is an RMG that satisfies \cref{assu:efficient} and $\fG'$ is the corresponding regularized MG constructed in \cref{thm:equiv_reward_Markov}. Then, $\fG'$ is a TPZS regularized MG with separate regularizers, $\Omega_{i,h}(s,\cdot) = \Omega_{i,i}^h(s,\cdot)$ for each player $i$ and stage $(s,h)$. Therefore, computing an MPRNE for $\fG$ can be done in polynomial time by solving $\fG'$.
\end{theorem}

The proof of \cref{thm:efficient-mg} is provided in \cref{sec:proof_thm_efficient-mg}. We immediately have the following result.

\begin{corollary}\label{cor:efficient-mg}
    RMGs satisfying \cref{assu:efficient} can be solved in polynomial time using any efficient planning or learning algorithm for TPZS regularized MGs.
\end{corollary}

We point out that the classes of uncertainty set we discussed in \cref{thm:example_matrix} do not all satisfy \cref{assu:efficient}. Specifically, case 1.\ concerning general ball uncertainty does not satisfy our decomposition. However, we show very similar uncertainty sets do satisfy the conditions and thus can be solved efficiently, as stated in the following theorem.

\begin{theorem}\label{thm:efficient-examples}
Consider a regularized MG with payoff matrix $\brno$ and the regularizer  $\Omega_{i,h}:\cS \times \Pi\rightarrow \R$ for each player $i\in[N]$ and state $s \in \cS$. 
\begin{enumerate}[leftmargin=*]
\item If $\Omega_{i,h}(s,\bpi_{h}) = \alpha_{i,s,h} \norm{\pi_{i,h}(s)}_p$ is the $p$-norm regularizer for each $i\in[N],h \in [H], s \in \cS$, then solving for MPNE of the regularized game is equivalent to solving for MPRNE of the robust game with ball constrained uncertainty set $\rset_{1,s,h}=\big\{ R_{1}\in\R^{\cA_{1}\times\cA_{2}}:\norm{R_{1}}_{\infty\to p}\leq\alpha_{1,s,h}\big\}$ and $\rset_{2,s,h}=\big\{ R_{2}\in\R^{\cA_{1}\times\cA_{2}}:\norm{R_{2}^{\top}}_{\infty\to p}\leq\alpha_{2,s,h}\big\}$.
\item The Markov game version of $(s,a)$-rectangular, policy-dependent reward uncertainty set from part 2.\ of \cref{thm:example_matrix}  carries over without any additional restrictions.
\end{enumerate}
In either case, both the RMG and regularized MG can be solved in polynomial time.
\end{theorem}

\begin{remark*}
    As with \cref{thm:example_matrix}, classical and common regularizers can be applied as special cases of \cref{thm:efficient-examples}. See \cref{sec:example_reward_matrix} for additional discussion.
\end{remark*}

\section{Markov Games with Transition Uncertainty}\label{sec:transition_robust}

Up to this point, we have focused on RMGs with reward uncertainty. In this section, we consider transition uncertainty.
Similar to before, 
we show that RMGs with transition uncertainty can be solved using regularized MG methods. One can
easily integrate the results of this section with those of \cref{sec:reward_robust} to characterize general RMGs with both reward and transition uncertainty. Proofs for the results in this section are deferred to \cref{sec:proof_section_transition}.

If $\fG$ is an RMG with $s$-rectangular transition uncertainty, then the uncertainty set takes the form $\uset=\pset\times\{\brno\}$. We can again derive a robust policy evaluation equation for $\fG$ similar to that in \cref{prop:reward_matrix} for reward uncertainty.

\begin{proposition}
\label{prop:transition_Markov} Consider a RMG $\fG$
with $s$-rectangular uncertainty set $\uset=\pset\times\{\brno\}$. Then for each product joint Markovian policy $\bpi\in\Pi$, the robust value function $\{V_{i,h}^{\bpi}\}_{h\in[H]}$
of each player $i\in[N]$ satisfies 
\begin{equation}\label{eq:transition_V_RBE}
\begin{split}
V_{i,h}^{\bpi}(s) =& \E_{\ba\sim\bpi_{h}(s)}\left[\rno_{i,h}(s,\ba)\right] -\sigma_{\pset_{s,h}}\left(-V_{i,h+1}^{\bpi}\bpi_{h}(s)^{\top}\right)
\end{split}
\end{equation}
where $[V_{i,h+1}^{\bpi}\bpi_{h}(s)^{\top}](s',\ba')=V_{i,h+1}^{\bpi}(s')\bpi_{h}(\ba'|s)$
for all $s'\in\cS,\ba'\in\cA$. 
\end{proposition}

Observe that \eqref{eq:transition_V_RBE} replaces the linear expected
future value function in a standard MG with a support function $\sigma_{\pset_{s,h}}(\cdot)$,
which depends on both the policy and future robust value. Depending
on the uncertainty set, the support function might involve a \emph{non-linear}
transformation of the policy and the value function.

To accommodate the non-linearity, we introduce a generalization of regularized Markov games where $\Omega:=\left(\Omega_{s,h}\right)_{s\in\cS,h\in[H]}$
is a finite set of policy-value regularization functions such that
for all $s\in\cS,h\in[H]$, $\Omega_{s,h}:\Delta(\cS)^{\cA}\times\R^{\cS}\times\Delta(\cA)\rightarrow\R$
satisfies that for each $P_{h}\in\Delta(\cS)^{\cA}$ and for each
$v\in\R^{\cS}$, $\Omega_{s,h}(P_{h},v,\cdot)$ is convex. Given a
joint policy $\bpi$ for a MG with $G=(P,\br)$, the general \emph{policy-value}
\emph{regularized }value function is defined recursively as follows: $\forall s\in\cS,h\in[H],$
\begin{align}
\VRT_{i,h}^{\bpi}(s,G) :=& \E_{\ba\sim\bpi_{h}(s)}\left[r_{i,h}(s,\ba)\right] \nonumber\\
&-\Omega_{s,h}\left(P_{h},-V_{i,h+1}^{\bpi},\bpi_{h}(s)\right).\label{eq:V_T_REMG}
\end{align}
Note that if $\Omega_{s,h}\left(P_{h},-v,\mu\right)=\dotp{P_{h},-v\mu^{\top}}=-\E_{\ba\sim\mu}\left[[P_{h}v](s,\ba)\right]$
for all $h\in[H]$, this regularized MG reduces to the standard MG. 

Viewing the support function $\sigma_{\pset_{s,h}}(\cdot)$ as the
policy-value regularizer, Proposition \ref{prop:transition_Markov}
implies that the RNE of robust Markov games with transition uncertainty
is equivalent to the NE of the regularized Markov game. 

\begin{theorem}
\label{thm:equiv_transition_Markov} 
Consider a RMG $\fG$ with $s$-rectangular uncertainty set $\uset=\pset\times\{\brno\}$. Consider the policy-value regularized MG $\fG' = (\cS,\{\cA_{i}\}_{i\in[N]},\pno,\brno,H,\Omega),$ where the regularizer
functions $\Omega:=\left(\Omega_{s,h}\right)_{s,h\in[H]}$ satisfies
$\Omega_{s,h}\left(\pno_{h},-v,\mu\right)=\sigma_{\pset_{s,h}}\left(-v\mu^{\top}\right),$
$\forall v\in\R^{\cS},\forall\mu\in\Delta(\cA)$. Then, $\bpi$ is an MPRNE for $\fG$ if and only if $\bpi$ is an MPNE for $\fG'$. 
\end{theorem}

Depending on the uncertainty set $\pset_{s,h}$, the support function
can be further simplified, leading to efficient computation
of regularized value functions. Here we give some examples
of transition uncertainty sets that are commonly considered in the literature,
including both $s$-rectangular and $(s,a)$-rectangular
sets. We refer readers to  \cref{sec:examples_transition} for the proof and discussion of additional examples.

\begin{corollary} \label{cor:ball_transition}
Consider a robust MG with uncertainty set $\uset=\pset\times\{\brno\},$ where $\pset=\times_{(s,h)\in\cS\times [H]}\pset_{s,h}$
satisfies ball constraints:
$\pset_{s,h}=\big\{ P\in\R^{\cS\times\cA}:\big\Vert P-\pno_{s,h}\big\Vert_{q^{*}\rightarrow p}\leq\beta_{s,h}\big\}$.
Then the equivalent policy-value regularized MG is associated with the
convex regularizer function $\Omega=\{\Omega_{s,h}\}_{s\in \cS,h\in[H]}$ such that $\forall s\in\cS,v\in\R^{\cS},\mu\in\Delta(\cA),$
\begin{align*}
\Omega_{s,h}&\left(\pno_{h},-v,\mu\right)  = \\
&-\E_{\ba\sim\mu}\left[[\pno_{h}v](s,\ba)\right]+\beta_{s,h}\left\Vert -v\right\Vert _{p}\left\Vert \mu^{\top}\right\Vert _{q}.\label{eq:norm_regularizer_transition}
\end{align*}  
\end{corollary}

We remark that for the policy-value regularized game in Corollary \ref{cor:ball_transition}, the corresponding value function
defined in \eqref{eq:V_T_REMG} reduces to 
\begin{align*}
  \VRT_{i,h}^{\bpi}(s,\gno)= &\E_{\ba\sim\bpi_{h}(s)}\big[\rno_{i,h}(s,\ba)+[\pno_{h}\VRT_{i,h+1}^{\bpi}](s,\ba)\big] \\
&-\beta_{s,h}\big\Vert -\VRT_{i,h+1}^{\bpi}\big\Vert _{p}\left\Vert \bpi_{h}(s)\right\Vert _{q}.
\end{align*}

Compared with standard MG, the value function involves an additional
term that penalizes the $\ell_{p}$-norm of future rewards and $\ell_{q}$-norm
of policy. Note that similar to the case with reward uncertainty, the size of the uncertainty set, i.e., the radius $\beta_{s,h}$,
determines the penalty factor.

When the transition uncertainty set is $(s,a)$-rectangular,
the regularizer functions for the equivalent policy-value regularized
MG admits simpler forms.
\begin{corollary}
\label{cor:SA_transition} Consider a robust MG with an uncertainty set
$\uset=\pset\times\{\brno\}$, where $\pset$ is $(s,a)$-rectangular
of the form $\pset=\times_{(s,\ba,h)\in\cS\times\cA\times [H]}\pset_{s,\ba,h}$,
with $\pset_{s,\ba,h}\subset\Delta(\cS)$ being compact and convex.
Then the equivalent policy-value regularized MG is associated with
convex regularizers $\Omega=\{\Omega_{s,h}\}$ such that: $\forall s\in\cS,v\in\R^{\cS},\mu\in\Delta(\cA),$
\[
\Omega_{s,h}\left(\pno_{h},-v,\mu\right)=\E_{\ba\sim\mu}\left[\sigma_{\pset_{s,\ba,h}}(-v)\right].
\]
\end{corollary}

Various $(s,a)$-rectangular transition uncertainty sets have
been considered for robust MDP \cite{shi2023robustRL,RMDP-RDP}. The uncertain transitions are typically of the form 
\[
\pset_{s,\ba,h}=\left\{ P\in\Delta(\cS):d(P,\pno_{s,\ba,h})\leq\beta_{s,\ba,h}\right\} ,
\]
where $\beta_{s,\ba,h}>0$ represents the level of uncertainty, and
$d(\cdot,\cdot)$ is a distance metric between two probability distributions.
Popular distance metrics include Total variation (TV) distance, KL
distance, Chi-square distance, and Wasserstein distance. For each case,
we can obtain equivalent policy-value regularizer functions. Here
we consider the TV distance and defer the discussion of other distance
metrics to  \Cref{sec:example_SA_transition}.

\paragraph{Example.} Consider a TV uncertainty set given by $\pset=\times_{(s,\ba,h)\in\cS\times\cA\times [H]}\pset_{s,\ba,h}^{\TV}$
with $\pset_{s,\ba,h}^{\TV}=\big\{ P\in\Delta(\cS):d_{\TV}(P,\pno_{s,\ba,h})\leq\beta_{s,\ba,h}\}.$
Here $d_{\TV}(\eta,\eta')=\frac{1}{2}\norm{\eta-\eta'}_{1}=\frac{1}{2}\sum_{s}|\eta(s)-\eta'(s)|$
for any $\eta,\eta'\in\Delta(\cS).$ Then the equivalent policy-value
regularized MG is associated with convex regularizers $\Omega=\{\Omega_{s,h}\}$
such that: $\forall s\in\cS,v\in\R^{\cS},\mu\in\Delta(\cA),$
\begin{align*}
&\Omega_{s,h}\left(\pno_{h},-v,\mu\right)=\E_{\ba\sim\mu}\big[-[\pno_{h}v](s,\ba) +\frac{\beta_{s,\ba,h}}{2} \cdot\\
&\min_{u\geq\boldsymbol{0}}\Big\{\max_{s'}\big(v(s')-u(s')\big)-\max_{s'}\big(v(s')-u(s')\big)\Big\}\Big].    
\end{align*}
In particular, the optimization in $\Omega_{s,h}$ is convex and
can be computed in time $\bigO(|\cS|\log|\cS|)$ \cite{RMDP-RDP}. Compared with standard MGs, the corresponding regularized value function
\eqref{eq:V_T_REMG} 
involves an additional
penalty that depends on the policy and future rewards. Similar to
the ball-constrained uncertainty set, the size of the uncertainty
set, i.e., radius $\beta_{s,\ba h}$ determines the penalty factor.
We remark that in general, the regularizer function might not include
the standard linear future value term $-[\pno_{h}v](s,\ba)$, as shown in ~\Cref{sec:example_SA_transition} for various transition uncertainty sets.

\textbf{Computational Hardness.} Lastly, we note that even for TPSZ RMGs with $(s,a)$-rectangular transition uncertainty and $H = 2$, computing a MPRNE is PPAD-hard. Thus, dealing with transition uncertainty is generally difficult. The proof uses transition uncertainty to simulate the same hard reward uncertainty instances derived in \cref{thm: zero-sum-hardness}.

\begin{theorem}\label{thm:transition-hardness}
    Even restricted to the class of $(s,a)$-rectangular uncertainty sets, computing an MPRNE of a TPZS RMG with transition uncertainty is PPAD-hard even for $S = H = 2$. 
\end{theorem}

\section{Conclusions}

In this work, we study RMGs with $s$-rectangular uncertainty. We show that RMGs can be solved using regularized MG algorithms. This reduction yields a planning algorithm for computing an MPRNE of an RMG. We also show that for many commonly used regularizers, the set of MPNE of the regularized game is equal to the set of MPRNE of RMGs with well-behaved uncertainty sets. This gives proof that regularization methods do produce robust policies.

However, we show even for two-player robust matrix games with $(s,a)$-rectangularity, computing an MPRNE is PPAD-hard. This illustrates that the reward uncertainty case is already challenging compared to the single-agent setting. Despite this, we show whenever the support function of the uncertainty set decomposes into a sum of two parts, one corresponding to each player's policy, then our constructed regularized game is zero-sum. Consequently, we can compute an MPRNE for two-player zero-sum RMGs with efficient player-decomposable reward uncertainty in polynomial time.

\section*{Impact Statement}
This paper presents work whose goal is to advance the field of game theory. There are many potential societal consequences of our work, none of which we feel must be specifically highlighted here.

\section*{Acknowledgements}
Q.\ Xie was supported in part by National Science Foundation Awards CNS-1955997 and EPCN-2339794. J.\ McMahan was supported in part by NSF grant 2023239.

\bibliography{robust}
\bibliographystyle{icml2024}

\newpage
\appendix
\onecolumn

\section{Useful Technical Results} \label{sec:technical}

In this section, we state some existing results that are useful in
our analysis.

\begin{theorem}[Fenchel-Rockafellar duality {\citep[Theorem 3.3.5]{borwein2010convex}}]
\label{thm:Fenchel_duality} Consider the problems: 
\begin{align*}
(P) &  &  & \min_{x}f(x)+g(Ax)\\
(D) &  &  & \max_{y}-f^{*}(-A^{*}y)-g^{*}(y)
\end{align*}
where $f:\cX\rightarrow\ER$ and $g:\mathcal{Y}\rightarrow\ER$ are
proper, closed convex functions, and $A:\cX\rightarrow\mathcal{Y}$
is a linear map. If the regularity condition $0\in\core(\dom(g)-A(\dom(f)))$
holds\footnote{Given $C\subset\R^{d}$, we say that $x\in\core(C)$ if for all $z\in\R^{d}$,
there exists a small enough $t\in\R$ such that $x+tz\in C$ \cite{borwein2010convex}.}, then the primal $(P)$ and dual $(D)$ optimal values are equal.
\end{theorem}

\begin{lemma}[Induced norm of rank-1 matrices]
\label{lem:norm_rank_1_matrix}For any vectors $u,v$, we have $\left\Vert uv^{\top}\right\Vert _{p\to q}=\left\Vert u\right\Vert _{q}\left\Vert v\right\Vert _{p^{*}}$.
\end{lemma}

\begin{proof}
By definition of matrix-induced norm, we have 
\begin{align*}
\left\Vert uv^{\top}\right\Vert _{p\to q} & =\sup_{x:x\neq0}\frac{\left\Vert uv^{\top}x\right\Vert _{q}}{\left\Vert x\right\Vert _{p}}\\
 & =\sup_{x:x\neq0}\frac{\left|v^{\top}x\right|\cdot\left\Vert u\right\Vert _{q}}{\left\Vert x\right\Vert _{p}} &  & \text{positive homogeneity of norms}\\
 & =\left\Vert u\right\Vert _{q}\cdot\sup_{x:x\neq0}\frac{\left|v^{\top}x\right|}{\left\Vert x\right\Vert _{p}}\\
 & =\left\Vert u\right\Vert _{q}\cdot\left\Vert v\right\Vert _{p^{*}}. &  & \text{definition of dual norm of }\left\Vert \cdot\right\Vert _{p}.
\end{align*}
\end{proof}
\begin{lemma}[Dual norm of induced matrix norm]
\label{lem:dual_norm_matrix_norm}For any matrix $R$, we have 
\[
\sup_{R:\left\Vert R\right\Vert _{p\to q}\le1}\left\langle R,xy^{\top}\right\rangle =\left\Vert x\right\Vert _{q^{*}}\left\Vert y\right\Vert _{p}=\left\Vert xy^{\top}\right\Vert _{p^{*}\to q}.
\]
\end{lemma}

\begin{proof}
By definition of dual norms, there exists a vector $u$ with $\left\Vert u\right\Vert _{q}=1$
such that $\left\langle u,x\right\rangle =\left\Vert x\right\Vert _{q^{*}}$.
Similarly, there exists $v$ with $\left\Vert v\right\Vert _{p^{*}}=1$
such that $\left\langle v,y\right\rangle =\left\Vert y\right\Vert _{p}$.
Then the matrix $R_{0}:=uv^{\top}$ satisfies
\[
\left\langle R_{0},xy^{\top}\right\rangle =\text{Tr}\left(R_{0}^{\top}xy^{\top}\right)=\text{Tr}\left(y^{\top}(vu^{\top})x\right)=\left\langle u,x\right\rangle \left\langle v,y\right\rangle =\left\Vert x\right\Vert _{q^{*}}\left\Vert y\right\Vert _{p}
\]
 and, by Lemma \ref{lem:norm_rank_1_matrix}, 
\[
\left\Vert R_{0}\right\Vert _{p\to q}=\left\Vert uv^{\top}\right\Vert _{p\to q}=\left\Vert u\right\Vert _{q}\left\Vert v\right\Vert _{p^{*}}=1.
\]
So 
\[
\sup_{\left\Vert R\right\Vert _{p\to q}\le1}\left\langle R,xy^{\top}\right\rangle \ge\left\langle R_{0},xy^{\top}\right\rangle =\left\Vert x\right\Vert _{q^{*}}\left\Vert y\right\Vert _{p}.
\]

On the other hand, for any $R$ with $\left\Vert R\right\Vert _{p\to q}\le1$,
we have
\[
\left\langle R,xy^{\top}\right\rangle =x^{\top}Ry\le\left\Vert x\right\Vert _{q^{*}}\left\Vert Ry\right\Vert _{q}\le\left\Vert x\right\Vert _{q^{*}}\left\Vert y\right\Vert _{p},
\]
hence $\sup_{\left\Vert R\right\Vert _{p\to q}\le1}\left\langle R,xy^{\top}\right\rangle \le\left\Vert x\right\Vert _{q^{*}}\left\Vert y\right\Vert _{p}$.
This proves the first equality in the lemma.

The second equality follows from Lemma \ref{lem:norm_rank_1_matrix}.
\end{proof}

\section{Properties of Robust Markov Games}\label{sec:proof_section_RMG}

Recall that given a joint policy $\bpi\in\Pi$ for a MG $G=(P,\br)$, the value
function and the state-action function for each player $i$ is defined
as: $\forall s\in\cS,\ba\in\cA,h\in[H],$
\begin{align}
\V_{i,h}^{\bpi}(s,G) &:=\E_{P}^{\bpi}\Big[\sum_{t=h}^{H}r_{i,t}(s_{t},\ba_{t})\mid s_{h}=s\Big],\label{eq:V_MG}\\
\Q_{i,h}^{\bpi}(s,\ba,G) &:=\E_{P}^{\bpi}\Big[\sum_{t=h}^{H}r_{i,t}(s_{t},\ba_{t})\mid s_{h}=s,\ba_{h}=\ba\Big],\label{eq:Q_MG}
\end{align}
where the expectation $\E_{P}^{\bpi}[\cdot]$ is taken with respect
to the trajectory $\{s_{h},a_{h},\br_{h}\}_{h\in[H]}$ induced by
the transition kernel $P$ and the joint policy $\bpi$, i.e., $\ba_{h}\sim\bpi_{h}(s_{h})$
and $s_{h+1}\sim P_{h}(\cdot|s_{h},\ba_{h})$. It is convenient
to set $\V_{i,H+1}^{\bpi}(s,G)\equiv\Q_{i,H+1}^{\bpi}(s,\ba,G)\equiv0$
for the terminal reward.

\subsection{Robust Markov Game Bellman Equation}

Similar to the Bellman equation for standard MG, we have the following robust Bellman equation. 
 
\begin{proposition}
\label{prop:Bellman}(Robust Bellman Equation). Under Assumption \ref{assu:rectangular},
for any joint policy $\bpi=\{\bpi_{h}\}_{h\in[H]}\in\Pi$, the following
Bellman equations hold: 
\begin{align}
V_{i,h}^{\bpi}(s) = &Q_{i,h}^{\bpi}\left(s,\bpi_{h}(\cdot\mid s)\right),\label{eq:RBE_V}\\
Q_{i,h}^{\bpi}(s,\mu) =&\inf_{\br_{h}\in\urset_{s,h}}\E_{\ba\sim\mu}\left[r_{i,h}(s,\ba)\right]+\inf_{P_{h}\in\pset_{s,h}}\E_{\ba\sim\mu}\left[[P_{h}V_{i,h+1}^{\bpi}](s,\ba)\right],\label{eq:RBE_Q}
\end{align}
where $\urset_{s,h}:=\times_{i\in[N]}\urset_{i,s,h}.$
\end{proposition}

\begin{proof}
\citet{RMARL-double-pessimism} established the robust Bellman Equation for robust
Markov game with $\cS\times\cA$-rectangular transition uncertainty
set. We extend the result to Markov games with general $\cS$-uncertainty
set, including both reward function uncertainty and transition uncertainty.
Throught out the proof, for each game model $G=(P,\br)$, we will
use $G^{h}:=\left(\{P_{t}\}_{t=h}^{H},\{\br_{t}\}_{t=h}^{H}\right)$
to denote the model parameters from step $h$ to the terminal step
$H$. With this notation, we note that for a standard Markov game
with model $G$, the value functions satisfy 
\begin{align*}
\V_{i,h}^{\bpi}(s,G) & =\V_{i,h}^{\bpi}(s,G^{h}),\\
\Q_{i,h}^{\bpi}(s,\ba,G) & =\Q_{i,h}^{\bpi}(s,\ba,G^{h}).
\end{align*}
To facilitate the proof, we introduce the following shorthands to
denote the uncertainty sets from step $h$ to the terminal step $H$:
\begin{align*}
\pset^{h} & :=\times_{(s,t)\in\cS\times\{t,t+1,\ldots,H\}}\pset_{s,t},\\
\uset^{r,h} & :=\times_{(i,s,t)\in[N]\times\cS\times\{t,t+1,\ldots,H\}}\uset_{i,s,t}^{r},\\
\uset^{h} & :=\pset^{h}\times\uset^{r,h}.
\end{align*}
 Let $\pset_{h}:=\times_{s\in\cS}\pset_{s,h}$ and $\urset_{h}:=\times_{(i,s)\in[N]\times\cS}\uset_{i,s,h}^{r}$.

We will prove the following stronger results via induction from step
$h=H$ to 1: given any $\bpi\in\Pi$, for each player $i\in[N]$,
there exists a game model $\gbei=(\pbei,\brbei)$ such that: (1) Robust
Bellman equations \eqref{eq:RBE_V}-\eqref{eq:RBE_Q} hold; (2) the robust
value functions satisfy 
\begin{align}
V_{i,h}^{\bpi}(s) & =\V_{i,h}^{\bpi}(s,\gbei^{h}),\qquad,\forall s\in\cS,\label{eq:V_RBE_g}\\
Q_{i,h}^{\bpi}(s,\mu) & =\E_{\ba\sim\mu}\left[\Q_{i,h}^{\bpi}(s,\ba,\gbei^{h})\right],\qquad\forall s\in\cS,\mu\in\Delta(\cA).\label{eq:Q_RBE_g}
\end{align}
\begin{enumerate}
\item (Base case): For $h=H$, by the definition of robust state-action
value function in \eqref{eq:Q_RMG}, it follows that \eqref{eq:RBE_Q}
holds. We also have 
\[
Q_{i,H}^{\bpi}(s,\mu)=\inf_{G\in\uset}\E_{\ba\sim\mu}\left[\Q_{i,H}^{\bpi}(s,\ba,G)\right]=\inf_{\br\in\uset^{r}}\E_{\ba\sim\mu}\left[r_{i,H}(s,\ba)\right]=\inf_{\br_{H}(s,\cdot)\in\uset_{s,H}^{r}}\E_{\ba\sim\mu}\left[r_{i,H}(s,\ba)\right].
\]
By Assumption \ref{assu:rectangular}, one can thus find a single
reward function $\brbei_{H}\in\uset_{s,H}^{r}=\times_{i}\urset_{i,s,H}$
such that for each $s\in\cS,$ 
\[
\brbei_{H}(s,\cdot)\in\arginf_{\br_{H}(s,\cdot)\in\uset_{s,H}^{r}}\E_{\ba\sim\mu}\left[r_{i,H}(s,\ba)\right].
\]
Therefore, 
\begin{equation}
Q_{i,H}^{\bpi}(s,\mu)=\E_{\ba\sim\mu}\left[\rbei_{i,H}(s,\ba)\right]=\E_{\ba\sim\mu}\left[\Q_{i,H}^{\bpi}(s,\ba,\gbei^{H})\right].\label{eq:QH_RBE}
\end{equation}
By definition of robust value function in \eqref{eq:V_RMG}, we have
\begin{align*}
V_{i,H}^{\bpi}(s) & =\inf_{G\in\uset}\V_{i,H}^{\bpi}(s,G)\overset{(I)}{=}\inf_{G\in\uset}\E_{\ba\sim\bpi_{H}(s)}\left[\Q_{i,H}^{\bpi}(s,\ba,G)\right]\overset{(II)}{=}Q_{i,H}^{\bpi}(s,\bpi_{H}(s))\overset{(III)}{=}\E_{\ba\sim\bpi_{H}(s)}\left[\Q_{i,H}^{\bpi}(s,\ba,\gbei^{H})\right]\\
 & \overset{(VI)}{=}\V_{i,H}^{\bpi}(s,\gbei^{H}),
\end{align*}
where equalities $(I)$ and $(VI)$ follow from the Bellman equations
\eqref{eq:V_MG}-\eqref{eq:Q_MG} of standard Markov game, equality $(II)$
follows from the definition of robust Q-function in \eqref{eq:Q_RMG},
and $(III)$ holds due to \eqref{eq:QH_RBE}. We complete the proof
of the base case. 
\item (Induction step): Now suppose that \eqref{eq:RBE_V}-\eqref{eq:RBE_Q}
and \eqref{eq:V_RBE_g}-\eqref{eq:Q_RBE_g} hold for all $t>h$. Thus
there exists $\gbei^{h+1}:=\left(\{\pbei_{t}\}_{t=h+1}^{H},\{\brbei_{t}\}_{t=h+1}^{H}\right)$
such that 
\begin{equation}
V_{i,h+1}^{\bpi}(s)=\V_{i,h+1}^{\bpi}(s,\gbei^{h+1}),\qquad,\forall s\in\cS.\label{eq:V_h1_RBE_g}
\end{equation}
By the definition of robust state-action value function in \eqref{eq:Q_RMG},
we have: $\forall s\in\cS,\mu\in\Delta(\cA)$, 
\begin{align}
Q_{i,h}^{\bpi}(s,\mu) & :=\inf_{G^{h}\in\uset^{h}}\E_{\ba\sim\mu}\left[\Q_{i,h}^{\bpi}(s,\ba,G^{h})\right]\nonumber \\
 & =\inf_{\{P_{t}\}_{t=h}^{H}\in\pset^{h},\{\br_{t}\}_{t=h}^{H}\in\uset^{r,h}}\E_{\ba\sim\mu}\left[r_{i,h}(s,\ba)+\E_{s'\sim P_{h}(\cdot|s,\ba)}\left[\V_{i,h+1}^{\bpi}(s',G^{h+1})\right]\right]\nonumber \\
 & \overset{(I)}{=}\inf_{\br_{h}(s,\cdot)\in\urset_{s,h}}\E_{\ba\sim\mu}\left[r_{i,h}(s,\ba)\right]+\inf_{\{P_{t}\}_{t=h}^{H}\in\pset^{h},\{\br_{t}\}_{t=h+1}^{H}\in\uset^{r,h+1}}\E_{\ba\sim\mu}\left[\E_{s'\sim P_{h}(\cdot|s,\ba)}\left[\V_{i,h+1}^{\bpi}(s',G^{h+1})\right]\right]\nonumber \\
 & \leq\inf_{\br_{h}(s,\cdot)\in\urset_{s,h}}\underbrace{\E_{\ba\sim\mu}\left[r_{i,h}(s,\ba)\right]}_{T_{1}}+\inf_{P_{h}\in\pset_{h}}\underbrace{\E_{\ba\sim\mu}\left[\E_{s'\sim P_{h}(\cdot|s,\ba)}\left[\V_{i,h+1}^{\bpi}(s',\gbei^{h+1})\right]\right]}_{T_{2}},\label{eq:Q_h_RBE_g}
\end{align}
where the equality $(I)$ follows from the rectangular uncertainty
set assumption and the fact that the first only depends on the reward
function at state $s$ and step $h$. We thus can find a single reward
function $\brbei_{h}\in\urset_{h}$that attains the minimum value
of the term $T_{1}$ for each state $s\in\cS$; we can also find a
single transition kernel $\pbei_{h}\in\pset_{h}$ that it attains
the minimum value of $T_{2}$. 

On the other hand, by \eqref{eq:V_h1_RBE_g} and the definition of robust
value function in \eqref{eq:V_RMG}, we have 
\begin{align}
Q_{i,h}^{\bpi}(s,\mu) & \leq\inf_{\br_{h}(s,\cdot)\in\urset_{s,h}}\E_{\ba\sim\mu}\left[r_{i,h}(s,\ba)\right]+\inf_{P_{h}\in\pset_{h}}\E_{\ba\sim\mu}\left[\E_{s'\sim P_{h}(\cdot|s,\ba)}\left[V_{i,h+1}^{\bpi}(s)\right]\right]\label{eq:Q_h_RBE}\\
 & =\inf_{\br_{h}(s,\cdot)\in\urset_{s,h}}\E_{\ba\sim\mu}\left[r_{i,h}(s,\ba)\right]\nonumber \\
 & \quad+\inf_{P_{h}\in\pset_{h}}\E_{\ba\sim\mu}\left[\E_{s'\sim P_{h}(\cdot|s,\ba)}\left[\inf_{\{P_{t}\}_{t=h+1}^{H}\in\pset^{h+1},\{\br_{t}\}_{t=h+1}^{H}\in\uset^{r,h+1}}\V_{i,h+1}^{\bpi}(s',G^{h+1})\right]\right]\nonumber \\
 & =\inf_{\br_{h}(s,\cdot)\in\urset_{s,h}}\E_{\ba\sim\mu}\left[r_{i,h}(s,\ba)\right]+\inf_{\{P_{t}\}_{t=h}^{H}\in\pset^{h},\{\br_{t}\}_{t=h+1}^{H}\in\uset^{r,h+1}}\E_{\ba\sim\mu}\left[\E_{s'\sim P_{h}(\cdot|s,\ba)}\left[\V_{i,h+1}^{\bpi}(s',G^{h+1})\right]\right]\nonumber \\
 & =Q_{i,h}^{\bpi}(s,\mu),\nonumber 
\end{align}
where the last equality follows from the definition of robust state-action
value function in \eqref{eq:Q_RMG}. Therefore, all the inequality
above are equalities. In particular, equation \eqref{eq:Q_h_RBE} proves
the robust Bellman equation \eqref{eq:RBE_Q} for step $h$. In addition,
from \eqref{eq:Q_h_RBE_g}, we have 
\begin{equation}
Q_{i,h}^{\bpi}(s,\mu)=\E_{\ba\sim\mu}\left[\rbei_{i,h}(s,\ba)\right]+\E_{\ba\sim\mu}\left[\E_{s'\sim\pbei_{h}(\cdot|s,\ba)}\left[\V_{i,h+1}^{\bpi}(s',\gbei^{h+1})\right]\right]=\E_{\ba\sim\mu}\left[\Q_{i,h}^{\bpi}(s,\ba,\gbei^{h})\right],\label{eq:Q_h_RBE_g_1}
\end{equation}
which proves \eqref{eq:Q_RBE_g} for step $h$. 

Next we will show that $V_{i,h}^{\bpi}(s)$ satisfies \eqref{eq:RBE_V}
and \eqref{eq:V_RBE_g}. By definition of robust value function in \eqref{eq:V_RMG},
we have 
\begin{align*}
V_{i,h}^{\bpi}(s) & =\inf_{\{P_{t}\}_{t=h}^{H}\in\pset^{h},\{\br_{t}\}_{t=h}^{H}\in\uset^{r,h}}\E_{P}^{\bpi}\left[\sum_{t=h}^{H}r_{i,h}(s_{t},\ba_{t})\mid s_{h}=s\right]\\
 & =\inf_{\{P_{t}\}_{t=h}^{H}\in\pset^{h},\{\br_{t}\}_{t=h}^{H}\in\uset^{r,h}}\E_{\ba\sim\bpi_{h}(s)}\left[\E_{P}^{\bpi}\left[\sum_{t=h}^{H}r_{i,h}(s_{t},\ba_{t})\mid s_{h}=s,\ba_{h}=\ba\right]\right]\\
 & \overset{(I)}{=}\inf_{\{P_{t}\}_{t=h}^{H}\in\pset^{h},\{\br_{t}\}_{t=h}^{H}\in\uset^{r,h}}\E_{\ba\sim\bpi_{h}(s)}\left[\Q_{i,h}^{\bpi}(s,\ba,G^{h})\right]\\
 & \leq\E_{\ba\sim\bpi_{h}(s)}\left[\Q_{i,h}^{\bpi}(s,\ba,\gbei^{h})\right]\\
 & \overset{(II)}{=}Q_{i,h}^{\bpi}(s,\bpi_{h}(s))\\
 & \overset{(III)}{=}\inf_{G^{h}\in\uset^{h}}\E_{\ba\sim\bpi_{h}(s)}\left[\Q_{i,h}^{\bpi}(s,\ba,G^{h})\right]\\
 & \overset{(VI)}{=}\inf_{G^{h}\in\uset^{h}}\V_{i,h}^{\bpi}(s,G^{h})\\
 & \overset{(V)}{=}V_{i,h}^{\bpi}(s),
\end{align*}
where $(I)$ follows from the definition of state-action value function
definition \eqref{eq:Q_MG}, $(II)$ follows from \eqref{eq:Q_h_RBE_g_1},
$(III)$ holds due to the definition of robust state-action value function
definition \eqref{eq:Q_RMG}, $(VI)$ is true from Bellman equation
of Markov game, and $(V)$ holds by definition in \eqref{eq:V_RMG}.
Therefore, the inequality above is equality. Note that $(II)$ proves
the Bellman equation \eqref{eq:RBE_V}. In addition, we have 
\[
V_{i,h}^{\bpi}(s)=\E_{\ba\sim\bpi_{h}(s)}\left[\Q_{i,h}^{\bpi}(s,\ba,\gbei^{h})\right]=\V_{i,h}^{\bpi}(s,\gbei^{h}).
\]
We complete the proof of step $h$. 

\end{enumerate}
This finishes the proof of Proposition \ref{prop:Bellman}.
\end{proof}

\section{Proof of Existence of Robust NE} \label{sec:proof_existence_rne}

In this section, we provide the proof of Theorem \ref{thm:Existence_RNE} on the existence of robust Nash equilibrium. We first consider matrix games in Section \ref{sec:existence_matrix} and then proceed to Markov games in Section \ref{sec:existence_markov}.

\subsection{Matrix Games with Reward Uncertainty}\label{sec:existence_matrix}

We first consider matrix games with reward uncertainty. In particular,
the reward uncertainty set is of the form $\uset=\brno+\rset$ with
$\rset=\times_{i\in[N]}\cR_{i}$, where $\rset_{i}\subset\R^{\cA}$
contains the possible reward functions for player $i$. Here we allow
the reward uncertainty sets to potentially depend on the players'
policy $\bpi=(\pi_{1},\ldots,\pi_{N})$ with $\pi_{i}\in\Delta(\cA_{i})$,
denoted by $\uset(\bpi)=\brno+\rset(\bpi)$ with $\rset(\bpi)=\times_{i\in[N]}\cR_{i}(\bpi).$
The existence of RNE has been established
for the setting where the uncertainty set is fixed and policy-independent \cite{RMG-original,RMARL-queuing,RM-robust-game-theory}. Here we extend
the result to more general uncertainty sets. 

Our proof uses Kakutani's Fixed Point Theorem \cite{kakutani1941generalization}. We first
state a relevant definition, followed by Kakutani's theorem below.
For a mapping from a closed, bounded, convex set $E$ in a Euclidean
space into the family of all closed, covex subsets of $E$, upper
semicontinuity is defined as follows. 

\begin{definition}
\label{def:usc} A point-to-set mapping $\phi:E\rightarrow2^{E}$
is called upper semicontinuous (u.s.c.) if 
\begin{align*}
y_{n} & \in\phi(x_{n}),\qquad n=1,2,3,\ldots\\
\lim_{n\rightarrow\infty}x_{n} & =x,\\
\lim_{n\rightarrow\infty}y_{n} & =y,
\end{align*}
imply that $y\in\phi(x)$. 
\end{definition}

\begin{theorem}[Kakutani's Fixed Point Theorem \cite{kakutani1941generalization}]
\label{thm:Kakutani} If
$E$ is a closed, bounded, and convex set in a Euclidean space, and
$\phi$ is an upper semicontinuous point-to-set mapping of $E$ into
the family of closed, convex subsets of $E$, then $\exists x\in E$
s.t. $x\in\phi(x)$. 
\end{theorem}

\subsubsection{Properties of Robust Value Functions}

To apply Kakutani's Fixed Point Theorem, we first establish some properties
of the worst-case expected payoff functions, i.e., robust value function,
defined as 
\[
F_{i}(\bpi)\triangleq\inf_{\br\in\uset(\bpi)}\pi_{i}^{\top}r_{i}\bpi_{-i}.
\]
Here we view $r_{i}\in\uset(\bpi)$ as a matrix in $\R^{\cA_{i}\times\cA_{-i}},$
$\pi_{i}\in\R^{\cA_{i}}$ as a column vector in and $\bpi_{-i}\in\R^{\cA_{-i}}$.
We begin by deriving an equivalent equation for the robust value function. 
\begin{proposition}
\label{prop:reward_matrix} For each player $i\in[N]$, given any
product joint policy $\bpi\in\Pi$ of all players, the robust
value for player $i$ satisfies 
\begin{align}
F_{i}(\bpi) & =\pi_{i}^{\top}\rno_{i}\bpi_{-i}-\sigma_{\rset_{i}(\bpi)}(-\pi_{i}\bpi_{-i}^{\top}),\label{eq:v_pi_in_sigma}
\end{align}
where $\sigma_{\cR_{i}(\bpi)}(\cdot)$ is the support function of
the reward uncertainty set $\cR_{i}(\bpi)$. 
    
\end{proposition}

\begin{proof}
Recall that the characteristic function $\delta_{\cR_{i}(\bpi)}:\Real^{\cA}\to\{0,\infty\}$
over a set $\cR_{i}(\bpi)\subseteq\Real^{\cA}$ is defined as $\delta_{\cR_{i}(\bpi)}(x)=0$
if $x\in\cR_{i}(\bpi)$ and $+\infty$ otherwise. The Legendre-Fenchel
transform of $\delta_{\cR_{i}(\bpi)}$, i.e., the support function
$\sigma_{\cR_{i}(\bpi)}:\Real^{\cA}\to(-\infty,+\infty]$, is defined
as 
\begin{equation}
\sigma_{\cR_{i}(\bpi)}(y):=\sup_{x\in\cR_{i}(\bpi)}\left\langle x,y\right\rangle =\sup_{x\in\Real^{\cA}}\left\{ \left\langle y,x\right\rangle -\delta_{\cR_{i}(\bpi)}(x)\right\} .\label{eq:def_sigma}
\end{equation}

By the form of the uncertainty set $\uset(\bpi)=\brno+\times\rset_{i}(\bpi)$,
given any $\bpi\in\Pi$, we have 
\begin{eqnarray*}
F_{i}(\bpi)=\inf_{\br\in\uset(\bpi)}\pi_{i}^{\top}r_{i}\bpi_{-i} & = & \inf_{x\in\rset_{i}(\bpi)}\pi_{i}^{\top}(\rno_{i}+x)\bpi_{-i}\\
 & = & \inf_{x\in\rset_{i}(\bpi)}\pi_{i}^{\top}x\bpi_{-i}+\pi_{i}^{\top}\rno_{i}\bpi_{-i}\\
 & = & \inf_{x\in\R^{\cA}}\{\pi_{i}^{\top}x\bpi_{-i}+\delta_{\cR_{i}(\bpi)}(x)\}+\pi_{i}^{\top}\rno_{i}\bpi_{-i}.
\end{eqnarray*}
We now proceed to apply Fenchel-Rockafellar duality theorem (Theorem
\ref{thm:Fenchel_duality}\textbf{)} to the minimization term. Fix
player $i$'s policy $\pi_{i}$ and all other players' policy $\bpi_{-i}$.
We define the function $f:\Real^{\cA}\to\R$ as $f(x)=\pi_{i}^{\top}x\bpi_{-i}$
for each $x\in\R^{\cA}$. Consider the identity mapping $\Id_{\cA}:\Real^{\cA}\to\Real^{\cA}$.
Then $\dom(f)=\Real^{\cA}$, $\dom(\delta_{\rset_{i}(\bpi)})=\rset_{i}(\bpi)$,
and thus $\core(\dom(\delta_{\rset_{i}(\bpi)})-\Id_{\cA}(\dom(f)))=\core(\rset_{i}(\bpi)-\R^{\cA})=\core(\R^{\cA})=\R^{\cA}.$
Note that $0\in\R^{\cA}$. We now can apply Fenchel-Rockafellar duality,
noting that $(\Id_{\cA})^{*}=\Id_{\cA}$ and $(\delta_{\rset_{i}(\bpi)})^{*}(y)=\sigma_{\cR_{i}(\bpi)}(y)$:
\[
\inf_{x\in\R^{\cA}}\{f(x)+\delta_{\cR_{i}(\bpi)}(x)\}=-\inf_{y\in\Real^{\cA}}\{f^{*}(-y)+\delta_{\rset_{i}(\bpi)}^{*}(y)\}=-\inf_{y\in\Real^{\cA}}\{f^{*}(-y)+\sigma_{\rset_{i}(\bpi)}(y)\}.
\]
We have
\begin{eqnarray*}
f^{*}(-y) & = & \sup_{x\in\R^{\cA}}\{\langle x,-y\rangle-\pi_{i}^{\top}x\bpi_{-i}\}\\
 & = & \sup_{x\in\R^{\cA}}\sum_{a_{i}\in\cA_{i}}\sum_{\ba_{-i}\in\cA_{-i}}x(a_{i},\ba_{-i})\left[-y(a_{i},\ba_{-i})-\pi_{i}(a_{i})\bpi_{-i}(\ba_{-i})\right]\\
 & = & \begin{cases}
0 & \text{if }-y(a_{i},\ba_{-i})-\pi_{i}(a_{i})\bpi_{-i}(\ba_{-i})=0\;\forall a_{i}\in\cA_{i},\ba_{-i}\in\cA_{-i}\\
+\infty & \text{otherwise}
\end{cases}
\end{eqnarray*}
Note that $y\in\R^{\cA}$ such that $y(a_{i},\ba_{-i})=-\pi_{i}(a_{i})\bpi_{-i}(\ba_{-i})$
is exactly the negative outer product of $\pi_{i}$ and $\bpi_{-i}$,
i.e., $y=-\pi_{i}\bpi_{-i}^{\top}$. Therefore, we have 
\[
F_{i}(\bpi)=-\inf_{y\in\Real^{\cA}}\{f^{*}(-y)+\sigma_{\rset_{i}(\bpi)}(y)\}+\pi_{i}^{\top}\rno_{i}\bpi_{-i}=\pi_{i}^{\top}\rno_{i}\bpi_{-i}-\sigma_{\rset_{i}(\bpi)}(-\pi_{i}\bpi_{-i}^{\top}).
\]
\end{proof}

\begin{lemma}
\label{lem:v_pi_continuous} Under Assumption \ref{assu:reward}, for each agent $i\in[N]$, $F_{i}(\bpi)$
is continuous on $\Pi.$ 
\end{lemma}

\begin{proof}
Here we focus on the case where the uncertainty set $\uset(\bpi)$
is well-defined for each $\bpi\in\int(\Pi)$. The proof for the setting
where $\uset(\bpi)$ is well-defined for all $\bpi\in\Pi$ follows
similarly. Recall that for an arbitrary $\bpi\in\int(\Pi)$, we have
$\sigma_{\cR_{i}(\bpi)}(-\pi_{i}\bpi_{-i}^{\top}):=\sup_{x\in\cR_{i}(\bpi)}\left\langle x,-\pi_{i}\bpi_{-i}^{\top}\right\rangle =\sup_{x\in\Real^{\cA}}\left\{ \left\langle -\pi_{i}\bpi_{-i}^{\top},x\right\rangle -\delta_{\cR_{i}(\bpi)}(x)\right\} $.
For each $x\in\cR_{i}(\bpi)$, by Assumption \ref{assu:reward}, we
have 
\[
\left\langle x,-\pi_{i}\bpi_{-i}^{\top}\right\rangle =-\E_{\ba\sim\bpi}\left[x(\ba)\right]\leq-L_{r}<+\infty.
\]
Therefore, the $\sup$ attains its maximum on the closed set $\cR_{i}(\bpi)$.
We denote the maximizer by $x^{*}(\bpi)\in\cR_{i}(\bpi)$. 

Let us consider $\epsilon>0$. Consider $\delta(\epsilon,\bpi)$ given
by 
\[
\bar{\delta}(\epsilon,\bpi)=\min\left\{ \frac{1}{2}\inf_{z\in\Bd(\Pi)}\left\Vert \bpi-z\right\Vert _{\infty},\delta\left(\frac{\epsilon}{3}\right),\frac{\min\{\epsilon,1\}}{6|\cA|\cdot\max\{r_{\max},M\}}\right\} ,
\]
where $M:=\max_{x\in\rset_{i}(\bpi)}\left\Vert x\right\Vert _{\infty}$.
For each $\bpi'\in\Pi$ such that $\left\Vert \bpi-\bpi'\right\Vert _{\infty}<\bar{\delta}(\epsilon,\bpi)$,
we have $\bpi'\in\int(\Pi)$. By Assumption \ref{assu:reward}, $D\big(\rset_{i}(\bpi),\rset_{i}(\bpi')\big)<\frac{\epsilon}{3}$.
Thus there exist $\hat{x}'\in\rset_{i}(\bpi')$ and $\hat{x}\in\rset_{i}(\bpi)$
such that $\left\Vert \hat{x}-x^{*}(\bpi')\right\Vert _{\infty}<\frac{\epsilon}{3}$
and $\left\Vert \hat{x}'-x^{*}(\bpi)\right\Vert _{\infty}<\frac{\epsilon}{3}$.

By Proposition \ref{prop:reward_matrix}, we have 
\begin{align*}
F_{i}(\bpi')-F_{i}(\bpi) & =(\pi'_{i})^{\top}\rno_{i}\bpi'_{-i}-\sup_{x\in\cR_{i}(\bpi')}\left\langle x,-\pi'_{i}(\bpi'_{-i})^{\top}\right\rangle -\pi_{i}^{\top}\rno_{i}\bpi_{-i}+\sup_{x\in\cR_{i}(\bpi)}\left\langle x,-\pi_{i}\bpi_{-i}^{\top}\right\rangle \\
 & \leq\left\langle \rno,\pi'_{i}(\bpi'_{-i})^{\top}-\pi_{i}\bpi_{-i}^{\top}\right\rangle -\left\langle \hat{x}',-\pi'_{i}(\bpi'_{-i})^{\top}\right\rangle +\left\langle x^{*}(\bpi),-\pi_{i}\bpi_{-i}^{\top}\right\rangle \\
 & =\left\langle \rno,\pi'_{i}(\bpi'_{-i})^{\top}-\pi_{i}\bpi_{-i}^{\top}\right\rangle -\left\langle \hat{x}'-x^{*}(\bpi),-\pi'_{i}(\bpi'_{-i})^{\top}\right\rangle +\left\langle x^{*}(\bpi),\pi'_{i}(\bpi'_{-i})^{\top}-\pi_{i}\bpi_{-i}^{\top}\right\rangle \\
 & \leq|\cA|\left\Vert \rno_{i}\right\Vert _{\infty}\left\Vert \pi'_{i}(\bpi'_{-i})^{\top}-\pi_{i}\bpi_{-i}^{\top}\right\Vert _{\infty}+\left\Vert \hat{x}'-x^{*}(\bpi)\right\Vert _{\infty}+|\cA|\left\Vert x^{*}(\bpi)\right\Vert _{\infty}\left\Vert \pi_{i}\bpi_{-i}^{\top}-\pi'_{i}(\bpi'_{-i})^{\top}\right\Vert _{\infty}\\
 & \leq|\cA|r_{\max}\cdot2\left\Vert \bpi-\bpi'\right\Vert _{\infty}+\left\Vert \hat{x}'-x^{*}(\bpi')\right\Vert _{\infty}+|\cA|\cdot M\cdot2\left\Vert \bpi-\bpi'\right\Vert _{\infty}\\
 & <\frac{\epsilon}{3}+\frac{\epsilon}{3}+\frac{\epsilon}{3}\\
 & =\epsilon.
\end{align*}

For the other direction, we have 
\begin{align*}
F_{i}(\bpi)-F_{i}(\bpi') & =\pi_{i}^{\top}\rno_{i}\bpi_{-i}-\sigma_{\rset_{i}(\bpi)}(-\pi_{i}\bpi_{-i}^{\top})-(\pi'_{i})^{\top}\rno_{i}\bpi'_{-i}+\sigma_{\rset_{i}(\bpi')}(-\pi'_{i}(\bpi'_{-i})^{\top})\\
 & =\pi_{i}^{\top}\rno_{i}\bpi_{-i}-\sup_{x\in\cR_{i}(\bpi)}\left\langle x,-\pi_{i}\bpi_{-i}^{\top}\right\rangle -(\pi'_{i})^{\top}\rno_{i}\bpi'_{-i}+\sup_{x\in\cR_{i}(\bpi')}\left\langle x,-\pi'_{i}(\bpi'_{-i})^{\top}\right\rangle \\
 & \leq\left\langle \rno,\pi_{i}\bpi_{-i}^{\top}-\pi'_{i}(\bpi'_{-i})^{\top}\right\rangle -\left\langle \hat{x},-\pi_{i}\bpi_{-i}^{\top}\right\rangle +\left\langle x^{*}(\bpi'),-\pi'_{i}(\bpi'_{-i})^{\top}\right\rangle \\
 & =\left\langle \rno,\pi_{i}\bpi_{-i}^{\top}-\pi'_{i}(\bpi'_{-i})^{\top}\right\rangle -\left\langle \hat{x},\pi'_{i}(\bpi'_{-i})^{\top}-\pi_{i}\bpi_{-i}^{\top}\right\rangle +\left\langle x^{*}(\bpi')-\hat{x},-\pi'_{i}(\bpi'_{-i})^{\top}\right\rangle \\
 & \leq|\cA|\left\Vert \rno_{i}\right\Vert _{\infty}\left\Vert \pi_{i}\bpi_{-i}^{\top}-\pi'_{i}(\bpi'_{-i})^{\top}\right\Vert _{\infty}+|\cA|\left\Vert \hat{x}\right\Vert _{\infty}\left\Vert \pi_{i}\bpi_{-i}^{\top}-\pi'_{i}(\bpi'_{-i})^{\top}\right\Vert _{\infty}+\left\Vert x^{*}(\bpi')-\hat{x}\right\Vert _{\infty}\\
 & \leq|\cA|r_{\max}\cdot2\left\Vert \bpi-\bpi'\right\Vert _{\infty}+|\cA|\cdot M\cdot2\left\Vert \bpi-\bpi'\right\Vert _{\infty}+\left\Vert \hat{x}-x^{*}(\bpi')\right\Vert \\
 & <\frac{\epsilon}{3}+\frac{\epsilon}{3}+\frac{\epsilon}{3}\\
 & =\epsilon.
\end{align*}

Therefore, $F_{i}(\cdot)$ is continuous in $\int(\Pi).$ By Assumption
\ref{assu:reward}, the support function $\sigma_{\rset_{i}(\bpi)}\big(-\pi_{i}\bpi_{-i}^{\top}\big)$
is continuous in the boundary of the compact set $\Pi$. From Proposition
\ref{prop:reward_matrix}, we have that $F_{i}(\cdot)$ is also continuous
in $\Bd(\Pi)$. Therefore, $F_{i}(\cdot)$ is continuous in $\Pi$. 
\end{proof}

\begin{lemma}
\label{lem:v_pi_concave}Under Assumptions \ref{assu:rectangular}
and \ref{assu:reward}, for each agent $i\in[N]$, $F_{i}(\bpi)$
is concave in $\pi_{i}$ given a fixed $\bpi_{-i}$. 
\end{lemma}

\begin{proof}
By Proposition \ref{prop:reward_matrix}, the concavity of $F_{i}(\bpi)$
in $\pi_{i}$ follows by the convexity of the support function $\sigma_{\rset_{i}(\bpi)}\big(-\pi_{i}\bpi_{-i}^{\top}\big)$. 
\end{proof}

\subsubsection{Existence of RNE}

\begin{theorem}
\label{thm:existence_matrix} Under Assumptions \ref{assu:rectangular}
and \ref{assu:reward}, the robust matrix game has an equilibrium.
\end{theorem}

\begin{proof}
We first construct a point-to-set mapping. Note that $\Pi$ is closed,
bounded and convex. We define $\phi:\Pi\rightarrow2^{\Pi}$ as 
\[
\phi(\bpi):=\left\{ z\in\Pi\mid z_{i}\in\argmax_{u_{i}\in\Delta(\cA_{i})}F_{i}(u_{i},\bpi_{-i}),i=1,\ldots,N\right\} .
\]
We next show that $\phi$ satisfies all the conditions in the Kakutani's
Fixed Point Theorem. 

By Lemma \ref{lem:v_pi_continuous}, $F_{i}(u_{i},\bpi_{-i})$ is
continuous. By Weierstrass' Theorem, the maximum of this continuous
function on a compact set $\Delta(\cA_{i})$ exists, i.e., $\argmax_{u_{i}\in\Delta(\cA_{i})}F_{i}(u_{i},\bpi_{-i})\neq\emptyset$.
We thus have $\phi(\bpi)\neq\emptyset$ for each $\bpi\in\Pi.$

Next, we show that $\phi(\bpi)$ is a convex set for each $\bpi\in\Pi$.
Suppose that $z,w\in\phi(\bpi).$ By the definition of $\phi$, for
each $i\in[N]$ and $\forall y_{i}\in\Delta(\cA_{i})$, we have 
\[
F_{i}(z_{i},\bpi_{-i})=F_{i}(w_{i},\bpi_{-i})\geq F_{i}(y_{i},\bpi_{-i}).
\]
Hence, for each $\lambda\in[0,1],$ we have 
\[
\lambda F_{i}(z_{i},\bpi_{-i})+(1-\lambda)F_{i}(w_{i},\bpi_{-i})\geq F_{i}(y_{i},\bpi_{-i}).
\]
By the concavity of $F_{i}$ from Lemma \ref{lem:v_pi_concave}, 
\[
F_{i}(\lambda z_{i}+(1-\lambda)w_{i},\bpi_{-i})\geq\lambda F_{i}(z_{i},\bpi_{-i})+(1-\lambda)F_{i}(w_{i},\bpi_{-i})\geq F_{i}(y_{i},\bpi_{-i}).
\]
Therefore, $\lambda z+(1-\lambda)w\in\phi(\bpi).$

We now show that $\phi$ is upper semi-continuous. Suppose that $\bpi^{n}\in\Pi$
with $\lim_{n\rightarrow\infty}\bpi^{n}=\bpi,$ and $z^{n}\in\phi(\bpi^{n})$
with $\lim_{n\rightarrow\infty}z^{n}=z$. By the definition of $\phi,$for
each $n$ and $i\in[N]$ and $\forall y_{i}\in\Delta(\cA_{i})$, we
have 
\[
F_{i}(z_{i}^{n},\bpi_{-i}^{n})\geq F_{i}(y_{i},\bpi_{-i}^{n}).
\]
By the continuity of $F_{i}$, if we take the limit on both sides,
we have 
\[
F_{i}(z_{i},\bpi_{-i})\geq F_{i}(y_{i},\bpi_{-i}).
\]
Hence $z\in\phi(\bpi).$ Therefore, $\phi$ is upper semi-continuous. 

Together, we show that $\phi$ satisfies all the conditions of Kakutani's
Fixed Point Theorem (Theorem \ref{thm:Kakutani}). Therefore, there
exists $\bpi\in\Pi,$ such that $\bpi\in\phi(\bpi)$. That is, there
exists an equilibrium in the robust matrix game. 
\end{proof}

\subsection{Markov Games: Proof of Theorem~\ref{thm:Existence_RNE}}\label{sec:existence_markov}
\begin{proof}

We define, 
\begin{equation}
\bpi_{h}^{\dagger}(s)\in NE(Q_{h}^{\bpine}(s,\cdot)).\label{equ: NE-helper}
\end{equation}
We will show that the above policy $\bpine$ is well defined. By definition
of NE and (\ref{eq:RBE_V}), this means that, 
\[
V_{i,h}^{\bpi^{\dagger}}(s)=Q_{i,h}^{\bpi^{\dagger}}(s,\bpine_{h}(s))=\sup_{u\in\Delta(\cA_{i})}Q_{i,h}^{\bpi^{\dagger}}(s,(u,\bpine_{-i,h}(s))).
\]

We show the stronger claim that for the policy $\bpine$ defined above,

\begin{equation}
V_{i,h}^{\bpine}(s)=\sup_{\tilde{\pi}_{i}\in\Pi_{i}}V_{i,h}^{(\tilde{\pi}_{i},\bpine_{-i})}(s),\quad\quad\forall h\in[H],s\in\cS,i\in[N],\label{eq:stronger}
\end{equation}
where $\Pi_{i}$ denotes the set of policies for agent $i$. We proceed
by induction on $h$. 
\begin{enumerate}
\item (\emph{Base Case}) Suppose that $h=H$. For each $s\in\cS$, by Theorem
\ref{thm:existence_matrix}, there exists an equilibrium for the robust
matrix game with reward uncertainty set $\times_{i\in[N]}\uset_{i,s,H}$.
Therefore, $\bpine_{H}(s)$ in (\ref{equ: NE-helper}) is well defined. We
also have,

\begin{align*}
V_{i,H}^{\bpine}(s) & =\sup_{u\in\Delta(\cA_{i})}Q_{i,H}^{\bpi^{\dagger}}\big(s,(u,\bpine_{-i,H}(s))\big)\\
 & =\sup_{u\in\Delta(\cA_{i})}\inf_{\br_{H}\in\urset_{s,H}}\E_{\ba\sim(u,\bpine_{-i,H}(s))}\left[r_{i,H}(s,\ba)\right]\\
 & =\sup_{\tilde{\pi}_{i}\in\Pi_{i}}\inf_{\br_{H}\in\urset_{s,H}}\E_{\ba\sim(\tilde{\pi}_{i,H}(s),\bpine_{-i,H}(s))}\left[r_{i,H}(s,\ba)\right]\\
 & =\sup_{\tilde{\pi}_{i}\in\Pi_{i}}V_{i,H}^{(\tilde{\pi}_{i},\bpine_{-i})}(s).
\end{align*}
The second equality follows from (\ref{eq:RBE_Q}) and the fourth
equality follows from (\ref{eq:RBE_V}). The third equality follows
since $\bpine$ is Markovian.

Overall, we see that (\ref{eq:stronger}) holds for $h=H$.
\item (\emph{Inductive Step}) Suppose that $h<H$. For any $i\in[N]$ and
$s\in\cS$, for each $\mu\in\Delta(\cA)$, we define 
\begin{align*}
F_{i,s,h}(\mu) & =Q_{i,h}^{\bpi^{\dagger}}\big(s,\mu\big)\\
 & =\inf_{\br_{h}\in\urset_{s,h}(\mu)}\inf_{P_{h}\in\pset_{s,h}}\E_{\ba\sim\mu}\left[r_{i,h}(s,\ba)+[P_{h}V_{i,h+1}^{\bpine}](s,\ba)\right].
\end{align*}
Following the same line of argument for the proof of Theorem \ref{thm:existence_matrix}
for a robust matrix game, we can show that $F_{i,s,h}(\mu)$ is continuous
and concave, and consequently, there exists an equilibrium for the
stage game with the expected payoff $\{F_{i,s,h}\}_{i\in[N]}.$ Therefore,
the NE $\bpine_{h}(s)$ in (\ref{equ: NE-helper}) is well defined.
We see that,

\begin{align*}
V_{i,h}^{\bpine}(s) & =\sup_{u\in\Delta(\cA_{i})}Q_{i,h}^{\bpi^{\dagger}}\big(s,(u,\bpine_{-i,h}(s))\big)\\
 & =\sup_{u\in\Delta(\cA_{i})}\inf_{\br_{h}\in\urset_{s,h}}\inf_{P_{h}\in\pset_{s,h}}\E_{\ba\sim(\tilde{\pi}_{i,h}(s),\bpine_{-i,h}(s))}\left[r_{i,h}(s,\ba)+[P_{h}V_{i,h+1}^{\bpine}](s,\ba)\right]\\
 & \leq\sup_{\tilde{\pi}_{i}\in\Pi_{i}}\inf_{\br_{h}\in\urset_{s,h}}\inf_{P_{h}\in\pset_{s,h}}\E_{\ba\sim(\tilde{\pi}_{i,h}(s),\bpine_{-i,h}(s))}\left[r_{i,h}(s,\ba)+[P_{h}V_{i,h+1}^{(\tilde{\pi}_{i},\bpine_{-i})}](s,\ba)\right]\\
 & =\sup_{\tilde{\pi}_{i,h}(s)\in\Delta(\cA_{i})}\inf_{\br_{h}\in\urset_{s,h}}\inf_{P_{h}\in\pset_{s,h}}\E_{\ba\sim(\tilde{\pi}_{i,h}(s),\bpine_{-i,h}(s))}\left[r_{i,h}(s,\ba)+\sup_{\tilde{\pi}_{i}\in\Pi_{i}}[P_{h}V_{i,h+1}^{(\tilde{\pi}_{i},\bpine_{-i})}](s,\ba)\right]\\
 & \leq\sup_{\tilde{\pi}_{i,h}(s)\in\Delta(\cA_{i})}\inf_{\br_{h}\in\urset_{s,h}}\inf_{P_{h}\in\pset_{s,h}}\E_{\ba\sim(\tilde{\pi}_{i,h}(s),\bpine_{-i,h}(s))}\left[r_{i,h}(s,\ba)+\mathbb{E}_{s'\sim P_{h}(s,a)}\left[\sup_{\tilde{\pi}_{i}\in\Pi_{i}}V_{i,h+1}^{(\tilde{\pi}_{i},\bpine_{-i})}(s')\right]\right]\\
 & =\sup_{\tilde{\pi}_{i,h}(s)\in\Delta(\cA_{i})}\inf_{\br_{h}\in\urset_{s,h}}\inf_{P_{h}\in\pset_{s,h}}\E_{\ba\sim(\tilde{\pi}_{i,h}(s),\bpine_{-i,h}(s))}\left[r_{i,h}(s,\ba)+\mathbb{E}_{s'\sim P_{h}(s,a)}\left[V_{i,h+1}^{\bpine}(s')\right]\right]\\
 & =\sup_{\tilde{\pi}_{i,h}(s)\in\Delta(\cA_{i})}Q_{i,h}^{\bpine}\big(s,(\tilde{\pi}_{i,h}(s),\bpine_{-i,h}(s))\big)\\
 & =V_{i,h}^{\bpine}(s).
\end{align*}
The first two equalities use (\ref{equ: NE-helper}) and (\ref{eq:RBE_Q}),
respectively. The first inequality follows by allowing player $i$
to also deviate at future steps. The second inequality uses Jensen's
inequality. The following equality follows from the induction hypothesis.
The last two equalities use (\ref{eq:RBE_Q}) and (\ref{equ: NE-helper})
respectively.

Next, since the starting and ending terms are the same, all inequalities
must be equalities. In particular, the first inequality is equality,
which is the relationship we wanted to show after applying (\ref{equ: NE-helper})
and (\ref{eq:RBE_Q}). Since this relationship holds for arbitrary
$s\in\cS$ and $i\in[N]$, we see that (\ref{eq:stronger}) holds
at time $h$.
\end{enumerate}
This completes the proof of (\ref{eq:stronger}). The fact that $\bpine$
is an NE then immediately follows from the $h=1$ case.

\end{proof}

\section{Analysis of Markov Games with Reward Uncertainty} \label{sec:proof_section_reward}

\subsection{Proof of Theorem \ref{thm:equiv_reward_matrix}}\label{sec:proof_thm_equiv_reward_matrix}

Let us first recall the definition of \emph{regularized value functions} for a regularized Markov game $(\cS,\{\cA_{i}\}_{i\in[N]},P,\br,H,\Omega),$ with $G=(P,\br)$. Given a joint policy $\bpi\in\Pi$, for each player $i$, $\forall s\in\cS,\ba\in\cA,h\in[H],$
\begin{align}\label{eq:V_REMG}
& \VR_{i,h}^{\bpi}(s,G)=\E_{P}^{\bpi}\Big[\sum_{t=h}^{H}r_{i,t}(s_{t},\ba_{t})-\Omega_{i,t}(\bpi_{t}(s_{t}))|s_h=s\Big],\\
& \QR_{i,h}^{\bpi}(s,\ba,G)=r_{i,h}(s,\ba) +\E_{P}^{\bpi}\Big[\sum_{t=h+1}^{H}\left(r_{i,t}(s_{t},\ba_{t})-\Omega_{i,t}(\bpi_{t}(s_{t}))\right)|s_{h}=s,\ba_{h}=\ba\Big]. \label{eq:BE_Q_REMG}
\end{align}

\begin{proof}
By the definition of the expected payoff $\VR_{i}^{\bpi}$ in \eqref{eq:V_REMG} for
each player $i$ in a regularized game, for each product joint policy $\bpi\in\Pi$, we have
\[
\VR_{i}^{\bpi}(\brno)=\pi_{i}^{\top}\rno_{i}\bpi_{-i}-\Omega_{i}(\bpi)=\pi_{i}^{\top}\rno_{i}\bpi_{-i}-\sigma_{\rset_{i}}(-\pi_{i}\bpi_{-i}^{\top}).
\]
For the robust game, by Proposition \ref{prop:reward_matrix}, for
each $\bpi\in\Pi$, the robust value of player $i$ satisfies:
\[
V_{i}^{\bpi}=\pi_{i}^{\top}\rno_{i}\bpi_{-i}-\sigma_{\rset_{i}}(-\pi_{i}\bpi_{-i}^{\top})=\VR_{i}^{\bpi}(\brno).
\]
Consider any RNE $\bpine$ of the robust game. By definition, we have
\[
\pine_{i}\in\argmax_{\pi_{i}\in\Delta(\cA_{i})}V_{i}^{\pi_{i}\times\bpine_{-i}}=\argmax_{\pi_{i}\in\Delta(\cA_{i})}\VR_{i}^{\pi_{i}\times\bpine_{-i}}(\brno),
\]
which implies that $\bpine$ is an NE of the regularized game. Following a
similar argument, we can conclude that any NE of the regularized game
is an RNE of the robust game.
\end{proof}

\subsection{Proof of Theorem \ref{thm:example_matrix}} \label{sec:proof_thm_example_matrix}

\begin{proof}
For ball constrained uncertainty set $\rset_{i}:=\left\{ R_{i}\in\R^{\cA_{i}\times\cA_{-i}}:\norm{R_{i}}_{q^{*}\rightarrow p}\leq\alpha_{i}\right\} $,
we have 
\begin{align*}
\sigma_{\rset_{i}}(-\pi_{i}\bpi_{-i}^{\top}) & =\sup_{R_{i}\in\rset_{i}}\dotp{R_{i},-\pi_{i}\bpi_{-i}^{\top}}\\
 & =\sup_{R_{i}:\norm{R_{i}}_{q^{*}\rightarrow p}\leq\alpha_{i}}\dotp{R_{i},-\pi_{i}\bpi_{-i}^{\top}}\\
 & =\alpha_{i}\sup_{R_{i}:\norm{R_{i}}_{q^{*}\rightarrow p}\leq1}\dotp{R_{i},-\pi_{i}\bpi_{-i}^{\top}}\\
 & =\alpha_{i}\norm{-\pi_{i}}_{p}\norm{\bpi_{-i}}_{q}\\
 & =\Omega_{i}(\bpi),
\end{align*}
where the second to last equality follows from Lemma \ref{lem:dual_norm_matrix_norm}
on the dual norm of matrix operator norm. The equivalence between
the robust game and the regularized games immediately follow from
Theorem \ref{thm:equiv_reward_matrix}.

For the $(s,a)$-rectangular policy-dependent uncertainty set, let $\rset_{i}({\bpi}):=\times_{\ba\in\cA}\rset_{i,\ba}({\bpi})$.
We have
\begin{align*}
\sigma_{\rset_{i}^{\bpi}}(-\pi_{i}\bpi_{-i}^{\top}) & =\sup_{R_{i}\in\rset_{i}({\bpi})}\dotp{R_{i},-\pi_{i}\bpi_{-i}^{\top}}\\
 & =-\sum_{a_{i}}\sum_{\ba_{-i}}\left[-\tau_{i}\omega_{i}(\pi_{i}(a_{i}))-g_{i}(\bpi_{-i}(\ba_{-i}))\right]\pi_{i}(a_{i})\bpi_{-i}(\ba_{-i})\\
 & =\tau_{i}\sum_{a_{i}}\pi_{i}(a_{i})\omega_{i}(\pi_{i}(a_{i}))+\sum_{\ba_{-i}}g_{i}(\bpi_{-i}(\ba_{-i}))\bpi_{-i}(\ba_{-i}).
\end{align*}
By Proposition \ref{prop:reward_matrix}, the robust best response
policy $\rbr(\bpi_{-i})$ for player $i$ is given by the following
optimization problem: 
\begin{align*}
 & \quad\arg\sup_{\pi_{i}\in\Delta(\cA_{i})}V_{i}^{\pi_{i}\times\bpi_{-i}}\\
 & =\arg\sup_{\pi_{i}\in\Delta(\cA_{i})}\left\{ \pi_{i}^{\top}\rno_{i}\bpi_{-i}-\sigma_{\rset_{i}^{\bpi}}(-\pi_{i}\bpi_{-i}^{\top})\right\} \\
 & =\arg\sup_{\pi_{i}\in\Delta(\cA_{i})}\left\{ \pi_{i}^{\top}\rno_{i}\bpi_{-i}-\tau_{i}\sum_{a_{i}}\pi_{i}(a_{i})\omega_{i}(\pi_{i}(a_{i}))-\sum_{\ba_{-i}}g_{i}(\bpi_{-i}(\ba_{-i}))\bpi_{-i}(\ba_{-i})\right\} \\
 & \equiv\arg\sup_{\pi_{i}\in\Delta(\cA_{i})}\left\{ \pi_{i}^{\top}\rno_{i}\bpi_{-i}-\tau_{i}\sum_{a_{i}}\pi_{i}(a_{i})\omega_{i}(\pi_{i}(a_{i}))\right\} \\
 & =\arg\sup_{\pi_{i}\in\Delta(\cA_{i})}\left\{ \pi_{i}^{\top}\rno_{i}\bpi_{-i}-\Omega_{i}(\bpi)\right\} 
\end{align*}
which gives the best response policy w.r.t. $\bpi_{-i}$ for the regularized
game with regularizer $\Omega=\{\Omega_{i}\}.$ Therefore, solving
the RNE of the robust game is equivalent to solving the NE of the
regularized normal-form game.
\end{proof}

\subsection{Examples of Game Regularization} \label{sec:example_reward_matrix}

As pointed out in Section \ref{sec:reward_matrix},
we can apply Theorem \ref{thm:example_matrix} to popular regularization schemes
in games, including negative Shannon entropy regularization, KL divergence
regularization, and Tsallis entropy regularization. Here we provide
details of the $\{\omega_{i}\}$ functions for the reward function
uncertainty set and two more examples of regularizers studied in games. 
\begin{itemize}
\item The negative Shannon entropy: $\Omega_{i}(\bpi)=\sum_{a_{i}\in\cA_{i}}\pi_{i}(a_{i})\log\pi_{i}(a_{i}).$
Thus we can define $\omega_{i}(\pi_{i}(a_{i})):=\log\pi_{i}(a_{i})$.
\item The KL divergence regularizer: $\Omega_{i}(\bpi)=\sum_{a_{i}\in\cA_{i}}\pi_{i}(a_{i})\log\frac{\pi_{i}(a_{i})}{\mu_{i}(a_{i})}=d_{\KL}(\pi_{i},\mu_{i})$,
where $\mu_{i}\in\Delta(\cA_{i})$ is a given distribution. We can
let $\omega_{i}(\pi_{i}(a_{i})):=\log\frac{\pi_{i}(a_{i})}{\mu_{i}(a_{i})}.$
\item The Tsallis entropy regularizer $\Omega_i(\bpi) = \frac{1}{2} \sum_{a_i \in \cA_{i}} (\pi_i(a_i)^2 - \pi_i(a_i)).$
Thus we can define $\omega_i(\pi_i(a_i)):=\frac{1}{2}(\pi_i(a_i) - 1).$
\item The Renyi (negative) entropy regularizer: $\Omega_{i}(\bpi)=-(1-q)^{-1}\log\left(\sum_{a_{i}\in\cA_{i}}\pi_{i}(a_{i})^{q}\right)$
for a given $q\in(0,1),$ and thus we can let (with a slight abuse of notation) $\omega_{i}(\pi_{i}, a_{i}):=-(1-q)^{-1}\log\left(\sum_{a'_{i}\in\cA_{i}}\pi_{i}(a'_{i})^{q}\right).$

\end{itemize}

\subsection{Proof of Proposition \ref{prop:reward_Markov}}\label{sec:proof_prop_reward_Markov}
\begin{proof}
    From Proposition \ref{prop:Bellman} we have that
    $$ V_{i,h}^{\bpi}(s) = \inf_{\br_{h}\in\urset_{s,h}} \E_{\ba\sim\bpi_{h}(s)}\left[r_{i,h}(s,\ba)\right]+\E_{\ba\sim\bpi_{h}(s)}\left[[\pno_{h}V_{i,h+1}^{\bpi}](s,\ba)\right].$$
    We proceed with backward induction over $h$.
    In the base case at step $h=H$, there is no further future transitions, thus $\forall s\in\cS$ and each $i\in [N]$
    \begin{eqnarray*}
        V_{i,H}^{\bpi}(s) &=& \inf_{\br_{H}\in\urset_{s,H}} \E_{\ba\sim\bpi_{H}(s)}\left[r_{i,H}(s,\ba)\right] \\
        &=& \E_{\ba\sim\bpi_{H}(s)}\left[\inf_{\br_{H}\in\urset_{s,H}} r_{i,H}(s,\ba)\right] \\
        &=& \E_{\ba\sim\bpi_{H}(s)} \left[  \rno_{i,H}(s,\ba)\right] -\sigma_{\rset_{i,s,H}}\left(-\pi_{i,H}(s)\bpi_{-i,H}^{\top}(s)\right)
    \end{eqnarray*}
    where the last equality holds from Proposition \ref{prop:reward_matrix}. 
    
    Now suppose that \eqref{eq:reward_Markov} holds for all steps $t>h$. Then $\forall s\in\cS$ and each $i\in [N]$, we have
    \begin{eqnarray*}
        V_{i,h}^{\bpi}(s) &=& \inf_{\br_{h}\in\urset_{s,h}} \E_{\ba\sim\bpi_{h}(s)}\left[r_{i,h}(s,\ba)\right]+\E_{\ba\sim\bpi_{h}(s)}\left[[\pno_{h}V_{i,h+1}^{\bpi}](s,\ba)\right] \\
        &=& \E_{\ba\sim\bpi_{h}(s)}\left[\inf_{\br_{h}\in\urset_{s,h}} r_{i,h}(s,\ba)\right]+\E_{\ba\sim\bpi_{h}(s)}\left[[\pno_{h}V_{i,h+1}^{\bpi}](s,\ba)\right] \\
        &=& \E_{\ba\sim\bpi_{h}(s)}\left[\rno_{i,h}(s,\ba)+[\pno_{h}V_{i,h+1}^{\bpi}](s,\ba)\right] -\sigma_{\rset_{i,s,h}}\left(-\pi_{i,h}(s)\bpi_{-i,h}^{\top}(s)\right),
    \end{eqnarray*}
    which completes the proof for the induction step. 
\end{proof}

\subsection{Proof of Theorem \ref{thm:equiv_reward_Markov}}\label{sec:proof_thm_equiv_reward_Markov}
\begin{proof}
    For the robust MG, by Proposition \ref{prop:reward_Markov}, for each
    $\bpi\in\Pi$, the robust value functions of player $i$ satisfy:
    $$V_{i,h}^{\bpi}(s)= \E_{\ba\sim\bpi_{h}(s)}\left[\rno_{i,h}(s,\ba)+[\pno_{h}V_{i,h+1}^{\bpi}](s,\ba)\right] -\sigma_{\rset_{i,s,h}}\left(-\pi_{i,h}(s)\bpi_{-i,h}^{\top}(s)\right).$$
    
    We can use backward induction from $h=H$ to 1 to show that 
    \begin{eqnarray*}
        V_{i,h}^{\bpi}(s)=\VR_{i,h}^{\bpi}(s,\gno),\quad\forall s\in\cS,\forall h\in[H].
    \end{eqnarray*}
    
    At the base step $h=H$, from Proposition \ref{prop:reward_Markov}, we have for $\forall s \in \cS$ and each $i\in [N]$
    \begin{eqnarray*}
        V_{i,H}^{\bpi}(s) &=& \E_{\ba\sim\bpi_{H}(s)}\left[\rno_{i,H}(s,\ba)\right] -\sigma_{\rset_{i,s,H}}\left(-\pi_{i,H}(s)\bpi_{-i,H}^{\top}(s)\right) \\
        &=& \E_{\ba\sim\bpi_{H}(s)}\left[\rno_{i,H}(s,\ba)\right]-\Omega_{i,H}(\bpi,s) \\
        &=& \VR_{i,H}^{\bpi}(s,\gno)
    \end{eqnarray*}
    
    Now assume that for $t>h$ that $$V_{i,t}^{\bpi}(s)=\VR_{i,t}^{\bpi}(s,\gno),\quad\forall s\in\cS, \forall i\in [N]$$
    and that $\bpine=(\pine_{1},\dots,\pine_{N})$ is an NE for the regularized game at steps $t>h$. 
    Then $\forall s \in \cS$ and each $i\in [N]$
    \begin{eqnarray*}
        V_{i,h}^{\bpi}(s) &=& \E_{\ba\sim\bpi_{h}(s)}\left[\rno_{i,h}(s,\ba)+[\pno_{h}V_{i,h+1}^{\bpi}](s,\ba)\right] -\sigma_{\rset_{i,s,h}}\left(-\pi_{i,h}(s)\bpi_{-i,h}^{\top}(s)\right) \\
        &=& \E_{\ba\sim\bpi_{h}(s)}\left[\rno_{i,h}(s,\ba)+[\pno_{h}\tilde{V}_{i,h+1}^{\bpi}](s,\ba)\right] -\Omega_{i,h}(\bpi,s) \\
        &=& \VR_{i,h}^{\bpi}(s,\gno)
    \end{eqnarray*}
    Thus $V_{i,h}^{\bpi}(s)=\VR_{i,h}^{\bpi}(s,\gno),\forall s\in\cS,\forall h\in[H],\forall i \in [N]$.
    
    Consider any RNE $\bpine$ of the robust MG. By definition, we have
    \[
    \pine_{i}\in\argmax_{\pi_{i}\in\Delta(\cA_{i})}V_{i,1}^{\pi_{i}\times\bpine_{-i}}(s_{1})=\argmax_{\pi_{i}\in\Delta(\cA_{i})}\VR_{i,1}^{\pi_{i}\times\bpine_{-i}}(s_{1},\gno),
    \]
    which implies that $\bpine$ is an NE of the regularized MG. Following a
    similar argument, we can conclude that any NE of the regularized MG
    is an RNE of the robust MG.
\end{proof}

\subsection{Examples of Regularized Markov Game and Equivalent Robust Markov Game} \label{sec:example_reward_Markov}

\begin{theorem}\label{thm:example_Markov}

Consider a regularized MG $\fG' = (\cS,\{\cA_{i}\}_{i\in[N]},\pno,\brno,H,\Omega)$ with regularizer functions $\Omega:=\left(\Omega_{i,h}\right)_{i\in[N],h\in[H]}.$
\begin{enumerate}
\item If $\Omega_{i,h}(s,\mu) = \alpha_{i,s,h} \norm{\mu_{i}}_p \norm{\mu_{-i}}_q$ is the $\ell_p$/$\ell_q$-norm regularizer for each $i\in[N],h \in [H], s \in \cS$ and $\forall \mu \in \Delta(\cA),$ then solving for MPNE of the regularized game $\fG'$ is equivalent to solving for MPRNE of the robust game $\fG$ with $s$-rectangular ball constrained reward uncertainty set $\uset^r = \brno+\times_{i,s,h}\rset_{i,s,h},$
$$\rset_{i,s,h}=\big\{ R_{i}\in\R^{\cA_{i}\times\cA_{-i}}:\norm{R_{i}}_{q^*\rightarrow p}\leq\alpha_{i,s,h}\big\},$$
where $q^*$ satisfies $\frac{1}{q^*}+\frac{1}{q}=1.$

\item If $\Omega_{i,h}$ is  decomposable with kernel $\omega$, i.e., $\Omega_{i,h}(s,\mu):=\tau_{i,s,h}\sum_{a_{i}\in \cA_i}\mu_{i}(a_i)\omega_{i,s,h}\big(\mu_{i}(a_{i})\big),$ $i\in[N],h \in [H], s \in \cS$, with $\tau_{i,s,h}\geq0$, and $\Omega_{i,h}(s,\mu)$ is convex in $\mu_i$ for each given $\mu_{-i}$. Then solving for NE of the regularized game $\fG'$ is equivalent to solving for RNE of robust game with $(s,a)$-rectangular policy-dependent uncertainty set $\uset({\bpi})=\brno+\times_{i,s,\ba,h}\rset_{i,s,\ba,h}({\bpi})$, where 
\begin{align*}
\rset_{i,s,\ba,h}({\bpi})=&\Big[\tau_{i,s,h}\omega_{i,s,h}\left(\pi_{i,h}(a_{i}|s)\right)+g_{i,s,h}\left(\bpi_{-i,h}(\ba_{-i}|s)\right), \\
&\;\;\wover_{i,s,h}\left(\pi_{i,h}(a_{i}|s)\right)+\gover_{i,s,h}(\bpi_{-i,h}(\ba_{-i}|s))\Big]\subset \R,
\end{align*}
with functions $\omega_{i,s,h},\wover_{i,s,h}: [0,1]\rightarrow\R$ and $g_{i,s,h},\gover_{i,s,h}: [0,1]\rightarrow\R$ are continuous.

\end{enumerate}

The proof follows easily from \cref{thm:example_matrix} using backward induction.

\end{theorem}

\subsubsection{Examples of Markov Game Regularization} \label{sec:example_reward_matrix_markov}

We can apply Theorem \ref{thm:example_Markov} above to popular regularization schemes
in Markov games. Below we provide four examples with the corresponding  $\{\omega_{i}\}$ functions for the reward function
uncertainty set. 
\begin{itemize}
\item The negative Shannon entropy: $\Omega_{i,h}(s,\mu)=  \sum_{a_{i}\in\cA_{i}}\mu_{i}(a_{i})\log\mu_{i}(a_{i}).$
Thus we can define $\omega_{i,s,h}(\pi_{i,h}(a_{i}|s)):=\log\pi_{i,h}(a_{i}|s)$.
\item The KL divergence regularizer: $\Omega_{i,h}(s,\mu)=\sum_{a_{i}\in\cA_{i}}\mu_{i}(a_{i})\log\frac{\mu_{i}(a_{i})}{\nu_{i}(a_{i})}=d_{\KL}(\mu_{i},\nu_{i})$,
where $\nu_{i}\in\Delta(\cA_{i})$ is a given distribution. 
We can let $\omega_{i,s,h}(\pi_{i,h}(a_{i}|s)):=\log\frac{\pi_{i,h}(a_{i}|s)}{\nu_{i}(a_{i})}.$
\item The Tsallis entropy regularizer $\Omega_{i,h}(s,\mu) = \frac{1}{2} \sum_{a_i \in \cA_{i}} (\mu_i(a_i)^2 - \mu_i(a_i)).$
Thus we can define $\omega_{i,s,h}(\pi_{i,h}(a_i|s)):=\frac{1}{2}(\pi_{i,h}(a_i|s) - 1).$
\item The Renyi (negative) entropy regularizer: $\Omega_{i,h}(s,\mu)=-(1-q)^{-1}\log\left(\sum_{a_{i}\in\cA_{i}}\mu_{i}(a_{i})^{q}\right)$
for a given $q\in(0,1).$ Thus we can let (with a slight abuse of notation) $\omega_{i,s,h}(\pi_{i,h}(\cdot|s),a_{i}):=-(1-q)^{-1}\log\left(\sum_{a'_{i}\in\cA_{i}}\pi_{i,h}(a'_{i}|s)^{q}\right).$

\end{itemize}

\section{Analysis of Robust Zero-Sum Markov Games} \label{sec:proof_section_efficient}

\subsection{Proof of \cref{thm: zero-sum-hardness}} \label{sec:proof_zero-sum-hardness}
\begin{proof}
    We present a poly-time reduction from the problem of computing an NE for a general-sum game to the problem of computing an RNE for a two-player zero-sum robust matrix game with $(s,a)$-rectangular reward uncertainty. Since computing an NE of a general-sum game is PPAD-hard, it then follows that computing an RNE for the aforementioned class of robust matrix games is also PPAD-hard. Let $(A,B)$ be an arbitrary general sum matrix game. WLOG we can further assume that $A, B \leq 0$. To construct the robust matrix game instance, we first define $\overline{r} = - \frac{A + B}{2}$ and $\underline{r} = \frac{A+B}{2}$. Then, we map $(A,B)$ to the robust matrix game $\fG$ defined by $\rno = A - \underline{r}$, $\rset_1 = \{r \in \mathbb{R}^{n_1 \times n_2} \mid \underline{r} \leq r \leq \overline{r}\}$, and $\rset_2 = - \rset_1$. 
    
    To prove the reduction is correct, we show that $\pi$ is an NE for $(A,B)$ if and only if $\pi$ is an RNE for $\fG$. First, we observe that for player $1$,
    \begin{align*}
        -\sigma_{\rset_1}(-\pi_1 \pi_2^{\top}) &= -\sup_{\underline{r} \leq R \leq \overline{r}} \langle R, -\pi_1 \pi_2^{\top} \rangle \\
        &= -\sup_{\underline{r} \leq R \leq \overline{r}} -\pi_1^{\top} R \pi_2 \\
        &= \inf_{\underline{r} \leq R \leq \overline{r}} \pi_1^{\top} R \pi_2 \\
        &= \pi_1^{\top} \underline{r} \pi_2.
    \end{align*}
    Using \cref{prop:reward_matrix}, we see the robust suboptimality gap for player $1$ under $\pi$ is exactly,
    \begin{equation*}
        D_1(\pi) = \max_{\pi_1'\in\Delta(\cA_1)}\big\{ \pi_1'^\top (\rno + \underline{r})\pi_2\big\} -(\pi_1^\top (\rno + \underline{r})\pi_2).
    \end{equation*}
    Similarly, we observe that $-\sigma_{-\rset_1}(-\pi_1 \pi_2^{\top}) = \pi_1^{\top} (-\overline{r}) \pi_2$. Thus, the robust suboptimality gap for player $2$ under $\pi$ is exactly,
    \begin{equation*}
        D_2(\pi) = \max_{\pi_2'\in\Delta(\cA_2)}\big\{\pi_1^\top(-\rno -\overline{r})\pi_2'\big\} - (\pi_1^\top (-\rno-\overline{r})\pi_2).
    \end{equation*}

    Putting these together, the RNE gap for $\pi$ is,
   \begin{align*}
        \rgap(\pi) &= \max_{\pi_1'\in\Delta(\cA_1)}\big\{ \pi_1'^\top (\rno+\underline{r})\pi_2\big\} +\max_{\pi_2'\in\Delta(\cA_2)}\big\{ \pi_1^\top (-\rno-\overline{r})\pi_2'\big\} -\pi_1^\top (\rno+\underline{r})\pi_2 - \pi_1^\top (-\rno-\overline{r}) \pi_2 \\
        &= \max_{\pi_1'\in\Delta(\cA_1)}\big\{ \pi_1'^\top A\pi_2\big\} +\max_{\pi_2'\in\Delta(\cA_2)}\big\{ \pi_1^\top B\pi_2'\big\} -\pi_1^\top A \pi_2 - \pi_1^\top B \pi_2.
    \end{align*}
    Observe this last term is exactly the NE gap for $(A,B)$. Thus, minimizing the optimality gaps for both games is equivalent which implies the set of RNEs for $\fG$ is exactly the set of NEs for $(A,B)$. In particular, this means that $\pi$ is an NE for $(A,B)$ if and only if $\pi$ is an RNE for $\fG$. Thus, the reduction is correct.

    Lastly, we note that $\fG$ can easily be computed in linear time in the size of $(A, B)$ just by computing the average and difference of matrices. Thus, the reduction can be done in polynomial time.

    By simply defining $\rset_1(\ba) = [\underline{r}(\ba), \overline{r}(\ba)]$ and $\rset_2(\ba) = [-\overline{r}(\ba), -\underline{r}(\ba)]$ the same proof applies to $(s,a)$-rectangularity.
\end{proof}

\subsection{Proof of \cref{lem:efficient-decomp}}\label{sec:proof_efficient-decomp}

\cref{prop:reward_matrix} implies that,
the robust value for player $i$ satisfies,

\begin{equation*}
V_{i}^{\bpi}=\pi_{i}^{\top}\rno_{i}\bpi_{-i}-\sigma_{\rset_{i}}(-\pi_{i}\bpi_{-i}^{\top})
\end{equation*}

Now, suppose that the characteristic function could be decomposed into $\sigma_{\rset_{i}}(-\pi_{i}\bpi_{-i}^{\top}) = \Omega_{i,i}(\pi_i) + \Omega_{i,-i}(\bpi_{-i})$. Then the robust suboptimality gap for player $i$ takes the form,

\begin{align*}
    V_{i}^{\dagger,\bpi_{-i}}(s)-V_{i}^{\bpi}(s)
    &= \max_{\pi_i' \in\Delta(\cA_1)} \left\{ {\pi_{i}'}^{\top}\rno_{i}\bpi_{-i}-\sigma_{\rset_{i}}(-\pi_{i}'\bpi_{-i}^{\top})\right\} - \left( \pi_{i}^{\top}\rno_{i}\bpi_{-i}-\sigma_{\rset_{i}}(-\pi_{i}\bpi_{-i}^{\top}) \right) \\
    &= \max_{\pi_i' \in\Delta(\cA_1)} \left\{ {\pi_{i}'}^{\top}\rno_{i}\bpi_{-i}-\Omega_{i,i}(\pi_i') -\Omega_{i,-i}(\bpi_{-i}) \right\} - \left( \pi_{i}^{\top}\rno_{i}\bpi_{-i}-\Omega_{i,i}(\pi_i) - \Omega_{i,-i}(\bpi_{-i}) \right) \\
    &= \max_{\pi_i' \in\Delta(\cA_1)} \left\{ {\pi_{i}'}^{\top}\rno_{i}\bpi_{-i}-\Omega_{i,i}(\pi_i') \right\} - \left( \pi_{i}^{\top}\rno_{i}\bpi_{-i}-\Omega_{i,i}(\pi_i) \right).
\end{align*}

Next, we define $\Omega_i(\pi_i) := \Omega_{i,i}(\pi_i)$.
We see the RNE gap takes the form,

\begin{align*}
    \rgap &= \max_{\pi_1' \in\Delta(\cA_1)} \left\{ {\pi_{1}'}^{\top}\rno\pi_{2}-\Omega_{1}(\pi_1') \right\} - \left( \pi_{1}^{\top}\rno_{1}\pi_{2}-\Omega_{1}(\pi_1) \right) \\
    &+ \max_{\pi_2' \in\Delta(\cA_2)} \left\{ {\pi_{1}}^{\top}(-\rno)\pi_{2}'-\Omega_{2}(\pi_2') \right\} - \left( \pi_{1}^{\top}(-\rno)\pi_{2}-\Omega_{2}(\pi_2) \right) \\
    &= \max_{\pi_1' \in\Delta(\cA_1)} \left\{ {\pi_{1}'}^{\top}\rno\pi_{2}-\Omega_{1}(\pi_1') + \Omega_2(\pi_2) \right\} \\
    &- \min_{\pi_2' \in\Delta(\cA_2)} \left\{ {\pi_{1}}^{\top}\rno\pi_{2}' - \Omega_1(\pi_1) +\Omega_{2}(\pi_2') \right\}
\end{align*}

This exactly matches the RNEGap for the zero-zum regularized game with regularization functions $\Omega_1$ and $\Omega_2$. Thus, solving the original two-player zero-sum robust game is equivalent to solving the two-player zero-sum regularized game with regularization functions $\Omega_1$ and $\Omega_2$. Furthermore, since $\Omega_{1,1}$ and $\Omega_{2,2}$ are strongly convex by assumption, we know this regularized game can be solved in polynomial time~\cite{efficient-regularizer}.

\subsection{Proof of \cref{thm:efficient-mg}}\label{sec:proof_thm_efficient-mg}

\begin{proof}
    The proof is nearly immediate given the proof of \cref{thm:equiv_reward_Markov}, \cref{lem:efficient-decomp}, and the definition of regularized MGs. For completeness, we give a formal proof below.
    We follow the proof of \cref{thm:equiv_reward_Markov} and show the constructed regularized game is TPZS.
    We proceed by backward induction on $h$. For the base case, we consider $h = H$. Fix any policy $\bpi$. For any $s \in \cS$ and $i \in [N]$, we know by \eqref{eq:reward_Markov},
    \begin{equation*}
        V_{i,H}^{\bpi}(s) = \E_{\ba\sim\bpi_{H}(s)}\left[\rno_{i,H}(s,\ba)\right]-\sigma_{\rset_{i,s,H}}\left(-\pi_{i,H}(s)\bpi_{-i,H}^{\top}(s)\right).
    \end{equation*}
    Since $\sigma_{\rset_{i,s,H}}\left(-\pi_{i,H}(s)\bpi_{-i,H}^{\top}(s)\right) = \Omega_{i,i}^H(\pi_{i,H}(s)) + \Omega_{i,-i}^H(\bpi_{-i,H}^{\top}(s))$ by \cref{assu:efficient}, \cref{lem:efficient-decomp} implies an NE for the robust stage game is equivalent to an NE for the corresponding TPZS regularized game. 

    For the inductive step, suppose that $h < H$. Fix any policy $\bpi$. For any $s \in \cS$ and $i \in [N]$, again we know by \eqref{eq:reward_Markov},
    \begin{equation*}
        V_{i,H}^{\bpi}(s) = \E_{\ba\sim\bpi_{H}(s)}\left[\rno_{i,h}(s,\ba) + [\pno_{h}V_{i,h+1}^{\bpi}](s,\ba)\right] -\sigma_{\rset_{i,s,h}}\left(-\pi_{i,h}(s)\bpi_{-i,h}^{\top}(s)\right).
    \end{equation*}
    From the proof of \cref{thm:equiv_reward_Markov}, we know that $[\pno_{h}V_{i,h+1}^{\bpi}](s,\ba) = [\pno_{h}\tilde{V}_{i,h+1}^{\bpi}](s,\ba)$. Furthermore, the induction hypothesis implies the future game is TPZS and so admits a unique NE value $\tilde{V}_{i,h+1}^{\bpi}(s') = \tilde{V}^*_{i,h+1}(s')$ for each $s' \in \cS$. Define $r'_{i,h}(s,\ba) := \rno_{i,h}(s,\ba) + [\pno_{h}\tilde{V}_{i,h+1}^{*}](s,\ba)$. Since $\sigma_{\rset_{i,s,h}}\left(-\pi_{i,h}(s)\bpi_{-i,h}^{\top}(s)\right) = \Omega_{i,i}^h(\pi_{i,h}(s)) + \Omega_{i,-i}^h(\bpi_{-i,h}^{\top}(s))$ by \cref{assu:efficient}, we see that solutions to the robust stage game correspond to a TPZS regularized game by \cref{lem:efficient-decomp}. Furthermore, this regularized game is exactly the regularized stage game of $\fG'$. This completes the proof.
\end{proof}

\subsection{Proof of \cref{thm:efficient-examples}}\label{sec:proof_thm_efficient-examples}

\begin{proof}
    The two cases of efficiently decomposable structures are
    \begin{enumerate}
        \item For the ball constrained uncertainty case $\rset_{1,s,h}=\big\{ R_{1}\in\R^{\cA_{1}\times\cA_{2}}:\norm{R_{1}}_{\infty\to p}\leq\alpha_{1,s,h}\big\}$ and $\rset_{2,s,h}=\big\{ R_{2}\in\R^{\cA_{1}\times\cA_{2}}:\norm{R_{2}^{\top}}_{\infty\to p}\leq\alpha_{2,s,h}\big\}$.

        We can show the decomposability for each player. 
        For player 1, we have
        \begin{align*}
            \sigma_{\rset_{1,s,h}}(-\pi_{1,h}(s)\pi_{2,h}^{\top}(s)) & =\sup_{R_{1}\in\rset_{1,s,h}} \dotp{R_{1},-\pi_{1,h}(s)\pi_{2,h}^{\top}(s)} \\
            & =\sup_{R_{1}:\norm{R_{1}}_{\infty\to p}\leq\alpha_{1,s,h}} \dotp{R_{1},-\pi_{1,h}(s)\pi_{2,h}^{\top}(s)} \\
            & =\alpha_{1,s,h} \sup_{R_{1}:\norm{R_{1}}_{\infty\to p}\leq 1} \dotp{R_{1},-\pi_{1,h}(s)\pi_{2,h}^{\top}(s)} \\
            & =\alpha_{1,s,h} \norm{-\pi_{1,h}(s)}_{p}
        \end{align*}
    
        Similarly, for player 2, we have
        \begin{align*}
            \sigma_{\rset_{2,s,h}}(-\pi_{1,h}(s)\pi_{2,h}^{\top}(s)) & =\sup_{R_{2}\in\rset_{2,s,h}} \dotp{R_{2},-\pi_{1,h}(s)\pi_{2,h}^{\top}(s)} \\
            & =\sup_{R_{2}:\norm{R_{2}^{\top}}_{\infty\to p}\leq\alpha_{2,s,h}} \dotp{R_{2},-\pi_{1,h}(s)\pi_{2,h}^{\top}(s)} \\
            & =\alpha_{2,s,h} \sup_{R_{2}:\norm{R_{2}^{\top}}_{\infty\to p}\leq 1} \dotp{R_{2},-\pi_{1,h}(s)\pi_{2,h}^{\top}(s)} \\
            & =\alpha_{2,s,h} \norm{-\pi_{2,h}(s)}_{p}
        \end{align*}
        where the last equality for both player views follows from Lemma \ref{lem:dual_norm_matrix_norm} on the dual norm of matrix operator norm. Thus we obtain decombosable $\Omega_{i,h}(s,\pi_{i,h}(s)) = \alpha_{i,s,h} \norm{\pi_{i,h}(s)}_p$.

        \item For the decomposable kernel case, 
        \begin{align*}
            \rset_{i,s,\ba,h}({\bpi}) =& \Big[\tau_{i,s,h}\omega_{i,s,h}\left(\pi_{i,h}(a_{i}|s)\right)+g_{i,s,h}\left(\bpi_{-i,h}(\ba_{-i}|s)\right), \wover_{i,s,h}\left(\pi_{i,h}(a_{i}|s)\right)+\gover_{i,s,h}(\bpi_{-i,h}(\ba_{-i}|s))\Big]\subset \R,
        \end{align*}
        with parameters $\tau_{i,s,h}\geq0$, functions $\omega_{i,s,h},\wover_{i,s,h}: [0,1]\rightarrow\R$ and $g_{i,s,h},\gover_{i,s,h}: [0,1]\rightarrow\R$.

        This is the same structure as in Theorem \ref{thm:example_Markov}, and thus we know the decomposable regularizer is $\Omega_{i,h}(s,\bpi)=\tau_{i,s,h}\sum_{a_{i}}\pi_{i,h}(a_i|s)\omega_{i,s,h}(\pi_{i,h}(a_{i}|s))$.
        
    \end{enumerate}

    With each case of $\Omega_{i,h}$ shown to be decomposable, we can proceed as in Theorem \ref{thm:example_Markov}. The MPRNE of the RMG, and MPNE of the corresponding regularized MG, are found by solving for the respective NE for each $s \in\cS$ at each $h \in [H]$. For either case, for each $s\in\cS$, we have from Theorem \ref{thm:example_matrix} that at each step $h = H$ we can efficiently solve the RNE through the corresponding regularized NE. Then we can proceed via backward induction to solve the NE for all $h$, as done in Theorem \ref{thm:equiv_reward_Markov}.
\end{proof}

\section{Analysis of Markov Games with Transition Uncertainty} \label{sec:proof_section_transition}

\subsection{Proof of Proposition \ref{prop:transition_Markov}} \label{sec:proof_prop_transition_Markov}

\begin{proof}
We will prove the proposition via an induction from step $h=H$ to
1. 

For the base case $h=H$, as there is no transition, \eqref{eq:transition_V_RBE}
holds by the definition of the robust value function. 

Now suppose that \eqref{eq:transition_V_RBE} holds for all $t\geq h+1$.
By Proposition \ref{prop:Bellman}, we have 
\begin{align*}
V_{i,h}^{\bpi}(s) & =\E_{\ba\sim\bpi_{h}(s)}\left[\rno_{i,h}(s,\ba)\right]+\inf_{P_{h}\in\pset_{s,h}}\E_{\ba\sim\bpi_{h}(s)}\left[[P_{h}V_{i,h+1}^{\bpi}](s,\ba)\right]\\
 & =\E_{\ba\sim\bpi_{h}(s)}\left[\rno_{i,h}(s,\ba)\right]+\inf_{P_{h}\in\pset_{s,h}}\E_{\ba\sim\bpi_{h}(s),s'\sim P_{h}(\cdot|s,\ba)}\left[V_{i,h+1}^{\bpi}(s')\right]\\
 & =\E_{\ba\sim\bpi_{h}(s)}\left[\rno_{i,h}(s,\ba)\right]+\inf_{P_{h}\in\pset_{s,h}}\dotp{P_{h},V_{i,h+1}^{\bpi}\bpi_{h}(s)^{\top}}\\
 & =\E_{\ba\sim\bpi_{h}(s)}\left[\rno_{i,h}(s,\ba)\right]-\sup_{P_{h}\in\pset_{s,h}}\dotp{P_{h},-V_{i,h+1}^{\bpi}\bpi_{h}(s)^{\top}}\\
 & =\E_{\ba\sim\bpi_{h}(s)}\left[\rno_{i,h}(s,\ba)\right]-\sigma_{\pset_{s,h}}\left(-V_{i,h+1}^{\bpi}\bpi_{h}(s)^{\top}\right).
\end{align*}
\end{proof}

\subsection{Proof of Theorem \ref{thm:equiv_transition_Markov}}\label{sec:proof_thm_equiv_transition}

\begin{proof}
For the robust MG, by Proposition \ref{prop:transition_Markov}, for each
$\bpi\in\Pi$, the robust value functions of player $i$ satisfy:
\[
V_{i,h}^{\bpi}(s)=\E_{\ba\sim\bpi_{h}(s)}\left[\rno_{i,h}(s,\ba)\right]-\sigma_{\pset_{s,h}}\left(-V_{i,h+1}^{\bpi}\bpi_{h}(s)^{\top}\right),\quad\forall s\in\cS,\forall h\in[H].
\]
We can use backward induction from $h=H$ to 1 to show that 
\[
V_{i,h}^{\bpi}(s)=\VRT_{i,h}^{\bpi}(s,\gno),\quad\forall s\in\cS,\forall h\in[H].
\]

Consider any RNE $\bpine$ of the robust MG. By definition, we have
\[
\pine_{i}\in\argmax_{\pi_{i}\in\Delta(\cA_{i})}V_{i,1}^{\pi_{i}\times\bpine_{-i}}(s_{1})=\argmax_{\pi_{i}\in\Delta(\cA_{i})}\VRT_{i,1}^{\pi_{i}\times\bpine_{-i}}(s_{1},\gno),
\]
which implies that $\bpine$ is an NE of the regularized MG. Following
similar argument, we can conclude that any NE of the regularized MG
is an RNE of the robust MG.
\end{proof}

\subsection{Examples of Transition Uncertainty Sets} \label{sec:examples_transition}

\subsubsection{Proof of Corollary \ref{cor:ball_transition}}\label{sec:proof_cor_ball_transiiton}

\begin{proof}
For ball constrained uncertainty set $\pset_{s,h}=\Big\{ P\in\R^{\cS\times\cA}:\norm{P-\pno_{s,h}}_{q^{*}\rightarrow p}\leq\beta_{s,h}\Big\},$
we have 
\begin{align*}
\sigma_{\pset_{s,h}}\left(-v\mu^{\top}\right) & =\sup_{P_{h}\in\pset_{s,h}}\dotp{P_{h},-v\mu^{\top}}\\
 & =\dotp{\pno_{h},-v\mu^{\top}}+\sup_{\norm{P'_{h}}_{p\rightarrow q^{*}}\leq\beta_{s,h}}\dotp{P'_{h},-v\mu^{\top}}\\
 & =-\E_{\ba\sim\mu}\left[[\pno_{h}v](s,\ba)\right]+\beta_{s,h}\sup_{\norm{P'_{h}}_{p\rightarrow q^{*}}\leq1}\dotp{P'_{h},-v\mu^{\top}}\\
 & \overset{(I)}{=}-\E_{\ba\sim\mu}\left[[\pno_{h}v](s,\ba)\right]+\beta_{s,h}\left\Vert -v\right\Vert _{p}\left\Vert \mu^{\top}\right\Vert _{q}\\
 & =\Omega_{s,h}\left(\pno_{h},-v,\mu\right)
\end{align*}
where equality $(I)$ follows from Lemma \ref{lem:dual_norm_matrix_norm}
on the dual norm of matrix operator norm. The equivalence between
the robust game and the regularized game immediately follows from
Theorem \ref{thm:equiv_transition_Markov}. 
\end{proof}

\subsubsection{Proof of Corollary \ref{cor:SA_transition}}\label{sec:proof_cor_SA_transition}

\begin{proof}
By Theorem \ref{thm:equiv_transition_Markov}, it is sufficient to
verify that the support function $\sigma_{\pset_{s,h}}(\cdot)$ is
equivalent to the specified regularizer function $\Omega_{s,h}$ for
each uncertainty set. When the transition uncertainty set is $(s,a)$-rectangular
with the form $\pset=\times_{(s,\ba,h)\in\cS\times\cA\times [H]}\pset_{s,\ba,h}$,
for each $s\in\cS,h\in[H]$, the support function $\sigma_{\pset_{s,h}}$
can be simplified. In particular, for each $v\in\R^{\cS},\mu\in\Delta(\cA)$,
we have 
\begin{align}
\sigma_{\pset_{s,h}}(-v\mu^{\top}) & =\sup_{P\in\pset_{s,h}}\dotp{P,-v\mu^{\top}}\nonumber \\
 & =\sup_{P=\{P_{\ba'}\}_{\ba'\in\cA}:P_{\ba'}\in\pset_{s,\ba',h}}\E_{\ba\sim\mu}\left[\dotp{P_{\ba},-v}\right]\nonumber \\
 & \overset{(I)}{=}\E_{\ba\sim\mu}\left[\sup_{P_{\ba}\in\pset_{s,\ba,h}}\dotp{P_{\ba},-v}\right]\nonumber \\
 & \overset{(II)}{=}\E_{\ba\sim\mu}\left[\sigma_{\pset_{s,\ba,h}}(-v)\right],\label{eq:support_product}
\end{align}
where $(I)$ holds due to the independence of uncertainty sets $\{\pset_{s,\ba,h}\}_{\ba\in\cA}$
and $(II)$ follows from the definition of support function. 
\end{proof}

\subsubsection{Examples of $\cS\times\cA$-rectangular Transition Uncertainty Sets}\label{sec:example_SA_transition}

We first formally introduce various distance metric for distributions.
\begin{enumerate}
\item Total variation distance: for any $\eta,\nu\in\Delta(\cS),$ $d_{\TV}(\eta,\nu)=\frac{1}{2}\norm{\eta-\nu}_{1}=\frac{1}{2}\sum_{s}|\eta(s)-\nu(s)|.$
\item Kullback-Leibler (KL) distance: for any $\eta,\nu\in\Delta(\cS),$
$d_{\KL}(\eta,\nu)=\sum_{s\in\cS}\eta(s)\log\frac{\eta(s)}{\nu(s)}$.
\item Chi-square distance: for any $\eta,\nu\in\Delta(\cS),$ $d_{\chi^2}(\eta,\nu)=\sum_{s}\frac{\left(\eta(s)-\nu(s)\right)^{2}}{\nu(s)}.$
\item Wasserstein distance: consider $p$-Wasserstein metric w.r.t. to a
metric $\rho(\cdot)$, i.e., for any $\eta,\nu\in\Delta(\cS),$ $d_{W_{p}}(\eta,\nu):=\left(\inf_{\gamma\sim\Gamma(\eta,\nu)}\E_{(x,y)\sim\gamma}\left[\rho(x,y)^{p}\right]\right)^{1/p},$
where $\Gamma(\eta,\nu)$ denotes the set of all couplings of $\eta$
and $\nu$. 
\end{enumerate}

\begin{corollary} \label{cor:example_SA_transition}
Consider a robust MG $(\cS,\{\cA_{i}\}_{i\in[N]},\pno,\brno,H,\uset)$
with uncertainty set $\uset=\pset\times\{\brno\}$, where $\pset$
is $(s,a)$-rectangular.
\begin{enumerate}
\item If $\pset$ is a KL uncertainty set given by $\pset=\times_{(s,\ba,h)\in\cS\times\cA\times [H]}\pset_{s,\ba,h}^{\KL}$
with $\pset_{s,\ba,h}^{\KL}=\big\{ P\in\Delta(\cS):d_{\KL}(P,\pno_{s,\ba,h})\leq\beta_{s,\ba,h}\big\},$
then the equivalent policy-value regularized MG $(\cS,\{\cA_{i}\}_{i\in[N]},\pno,\brno,H,\Omega)$
is associated with the following convex regularizer: $\forall s\in\cS,v\in\R^{\cS},\mu\in\Delta(\cA),$
\begin{equation}
\Omega_{s,h}\left(\pno_{h},-v,\mu\right)=\E_{\ba\sim\mu}\left[\min_{\lambda\geq0}\bigg\{\beta_{s,\ba,h}\lambda+\lambda\log\Big((\pno_{s,\ba,h})^{\top}\exp\big(-\frac{v}{\lambda}\big)\Big)\bigg\}\right].\label{eq:regularizer_P_KL}
\end{equation}
\item If $\pset$ is a Chi-square uncertainty set given by $\pset=\times_{(s,\ba,h)\in\cS\times\cA\times [H]}\pset_{s,\ba,h}^{\chi^{2}}$
with $\pset_{s,\ba,h}^{\chi^{2}}=\big\{ P\in\Delta(\cS):d_{\chi^{2}}(P,\pno_{s,\ba,h})\leq\beta_{s,\ba,h}\big\},$
then the equivalent policy-value regularized MG $(\cS,\{\cA_{i}\}_{i\in[N]},\pno,\brno,H,\Omega)$
is associated with the following convex regularizer: $\forall s\in\cS,v\in\R^{\cS},\mu\in\Delta(\cA),$
\begin{equation}
\Omega_{s,h}\left(\pno_{h},-v,\mu\right)=\E_{\ba\sim\mu}\left[\max_{u\geq\boldsymbol{0}}\Big\{(\pno_{s,\ba,h})^{\top}(v-u)-\sqrt{\beta_{s,\ba,h}\Var_{\pno_{s,\ba,h}}(v-u)}\Big\}\right];\label{eq:regularizer_P_Chi}
\end{equation}
\item If $\pset$ is a total variation uncertainty set given by $\pset=\times_{(s,\ba,h)\in\cS\times\cA\times [H]}\pset_{s,\ba,h}^{\TV}$
with $\pset_{s,\ba,h}^{\TV}=\big\{ P\in\Delta(\cS):d_{\TV}(P,\pno_{s,\ba,h})\leq\beta_{s,\ba,h}\},$
then the equivalent policy-value regularized MG $(\cS,\{\cA_{i}\}_{i\in[N]},\pno,\brno,H,\Omega)$
is associated with the following convex regularizer: $\forall s\in\cS,v\in\R^{\cS},\mu\in\Delta(\cA),$
\begin{equation}
\Omega_{s,h}\left(\pno_{h},-v,\mu\right)=\E_{\ba\sim\mu}\left[-[\pno_{h}v](s,\ba)+\frac{\beta_{s,\ba,h}}{2}\min_{u\geq\boldsymbol{0}}\bigg\{\max_{s'}\big(v(s')-u(s')\big)-\max_{s}\big(v(s)-u(s)\big)\bigg\}\right];\label{eq:regularizer_P_TV}
\end{equation}
\item Wasserstein uncertainty set: If the uncertainty set $\pset$ is $\cS\times\cA$-rectangular
given by $\pset=\times_{(s,\ba,h)\in\cS\times\cA\times [H]}\pset_{s,\ba,h}^{W_{p}}$
with $\pset_{s,\ba,h}^{W_{p}}=\big\{ P\in\Delta(\cS):d_{W_{p}}(P,\pno_{s,\ba,h})\leq\beta_{s,\ba,h}\big\},$
\begin{equation}
\Omega_{s,h}\left(\pno_{h},-v,\mu\right)=\E_{\ba\sim\mu}\left[\inf_{\lambda\geq0}\bigg\{\lambda\beta_{s,\ba,h}+\E_{\tilde{s}\sim\pno_{s,\ba,h}}\bigg[\sup_{s'\in\cS}\big\{-v(s')-\lambda\rho(\tilde{s},s')\big\}\bigg]\bigg\}\right].\label{eq:regularizer_P_Wasser}
\end{equation}
\end{enumerate}
\end{corollary}
\begin{proof} By Corollary \ref{cor:SA_transition}, it is sufficient to analyze the  support function $\sigma_{\pset_{s,\ba,h}}$ for each
case.
\begin{enumerate}
\item KL uncertainty set: From the strong duality result on KL constrained
set \citep[Lemma 4.1]{RMDP-RDP}
the optimization problem in the support function $\sigma_{\pset_{s,\ba,h}}(-v)$
is equivalent to 
\begin{equation}
\min_{\lambda\geq0}\bigg\{\beta_{s,\ba,h}\lambda+\lambda\log\Big((\pno_{s,\ba,h})^{\top}\exp\big(-\frac{v}{\lambda}\big)\Big)\bigg\},\label{eq:KL_dual}
\end{equation}
which is convex in $\lambda$ and can be solved efficiently. By \eqref{eq:support_product},
we have $\sigma_{\pset_{s,h}}(-v\mu^{\top})=\Omega_{s,h}\left(\pno_{h},-v,\mu\right)$
with $\Omega_{s,h}$ defined in \eqref{eq:regularizer_P_KL}.
\item Chi-square uncertainty set: From the strong duality result on $\chi^{2}$-distance
constrained set \citep[Lemma 4.2]{RMDP-RDP}, the optimization
problem in the support function $\sigma_{\pset_{s,\ba,h}}(-v)$ is
equivalent to 
\begin{equation}
\min_{u\geq\boldsymbol{0}}\left\{ -(\pno_{s,\ba,h})^{\top}(v-u)+\sqrt{\beta_{s,\ba,h}\Var_{\pno_{s,\ba,h}}\left(v-u\right)}\right\} ,\label{eq:Chi_dual}
\end{equation}
where $\Var_{\pno_{s,\ba,h}}\left(v-u\right)=(\pno_{s,\ba,h})^{\top}(v-u)^{2}-\left((\pno_{s,\ba,h})^{\top}(v-\mu)\right)^{2}$
and the convex optimization problem \eqref{eq:Chi_dual} can be solved
with complexity $\bigO(|\cS|\log|\cS|)$. By \eqref{eq:support_product},
we have $\sigma_{\pset_{s,h}}(-v\mu^{\top})=\Omega_{s,h}\left(\pno_{h},-v,\mu\right)$
with $\Omega_{s,h}$ defined in \eqref{eq:regularizer_P_Chi}.
\item TV uncertainty set: From the strong duality result on TV-constrained
set \citep[Lemma 4.3]{RMDP-RDP}, the optimization problem in
the support function $\sigma_{\pset_{s,\ba,h}}(-v)$ is equivalent
to 
\begin{equation}
-(\pno_{s,\ba,h})^{\top}v+\frac{\beta_{s,\ba,h}}{2}\min_{u\geq\boldsymbol{0}}\left\{ \max_{s'}\left(v(s')-u(s')\right)-\max_{s'}\left(v(s')-u(s')\right)\right\} ,\label{eq:TV_dual}
\end{equation}
which can be solved with complexity $\bigO(|\cS|\log|\cS|)$. By \eqref{eq:support_product},
we have $\sigma_{\pset_{s,h}}(-v\mu^{\top})=\Omega_{s,h}\left(\pno_{h},-v,\mu\right)$
with $\Omega_{s,h}$ defined in \eqref{eq:regularizer_P_TV}.
\item Wasserstein uncertainty set: From the strong duality result \citep[Theorem 1]{blanchet2019quantifying}, it holds that 
\[
\sigma_{\pset_{s,\ba,h}}(-v)=\sup_{P\in\pset_{s,\ba,h}^{W_{p}}}\E_{P}[-v]=\inf_{\lambda\geq0}\left\{ \lambda\beta_{s,\ba,h}+\E_{\tilde{s}\sim\pno_{s,\ba,h}}\left[\sup_{s'\in\cS}\left\{ -v(s')-\lambda\rho(\tilde{s},s')\right\} \right]\right\} ,
\]
which yields the equivalent regularizer function $\Omega_{s,h}$ defined
in \eqref{eq:regularizer_P_TV}.
\end{enumerate}
\end{proof}

\subsection{Proof of \cref{thm:transition-hardness}}\label{sec:proof_thm_transition_hardness}
\begin{proof}

Given a general-sum game $(A, B)$ we construct a transition-uncertain RMG that recovers the same properties as the reward-uncertain RMG from the proof of \cref{thm: zero-sum-hardness}. Then, the proof of hardness follows exactly as before. 

We define $\underline{r}$ and $\overline{r}$ exactly as before. Then, we define the nominal reward model by
$\rno_{1}(s_1,\ba) = A(\ba) - \underline{r}(\ba)$, $\rno_{2}(s_1,\ba) = \underline{r}(\ba)$, and $\rno_{2}(s_2, \ba) = \overline{r}(\ba)$. Here, we use the subscript to denote the time step, not the player number since the game is zero-sum. Lastly, we define $\pset_{s, a} = \Delta(\cS)$ to allow all possible transitions. Note, that we only need to define the transition uncertainty for the first step since $H = 2$. We assume the start state is $s_1$ so let $\rno_{2}(s_1,\ba)$ be arbitrary.

It is then clear that the worst model for the first player deterministically sends it to state $s_1$, which yields a reward of $\pi_1^{\top}(\rno + \underline{r})\pi_2$ as before. Similarly, the worst model for the second player deterministically sends it to state $s_2$ which yields a reward of $\pi_1^{\top}(-\rno -\underline{r})\pi_2$ as before. Thus, the earlier proof then applies and shows the problem is PPAD-hard to solve.

\end{proof}

\end{document}